\documentclass[a4paper,12pt]{article}
\usepackage{amsmath,amsfonts,amssymb,amsthm,latexsym}
\usepackage[numbers,sort&compress]{natbib}
\usepackage{mathrsfs}
\usepackage{stmaryrd}
\usepackage{amscd}
\usepackage{mathtools}
\numberwithin{equation}{section}
\usepackage[latin1]{inputenc}
\usepackage{lmodern}
\usepackage[T1]{fontenc}
\usepackage{dsfont}
\usepackage{verbatim}
\usepackage{enumerate}
\usepackage{xcolor}
\usepackage{tikz}
\usepackage[pdftex,colorlinks,bookmarks=true,unicode=true,pdftitle={two-arc study},pdfauthor={Cheng WAN}]{hyperref}
\usepackage{ifpdf}
\usepackage{url}

\usepackage[scale=0.8]{geometry}
\newtheorem{theorem}{Theorem}[section]
\newtheorem{lemma}[theorem]{Lemma}

\newtheorem{corollary}[theorem]{Corollary}

\newtheorem{assumption}{Assumption}
\theoremstyle{definition}

\theoremstyle{remark}

\hypersetup{
    linktocpage=true,
    citecolor=blue,    
    linkcolor=red    
}

\def\nit{\mathbb N}

\def\cS{{\mathcal S}}

\usetikzlibrary{positioning,shadows,arrows,shapes}

\title{Strategic decentralization in binary choice \\composite congestion games\footnote{11 September 2015}}

\author{Cheng Wan\footnote{\textit{email:} \nolinkurl{cheng.wan@economics.ox.ac.uk}  \textit{Tel.:} +44 (0)1865 278692} \\ {\small Department of Economics; Nuffield College, University of Oxford}\\ {\small New Road, OX1 1NF, Oxford, United Kingdom} }
 \date{ }
 
\begin{document}
\maketitle
 \begin{abstract}
This paper studies strategic decentralization in binary choice composite network congestion games. A player decentralizes  if she lets some autonomous agents to decide respectively how to send different parts of her stock from the origin to the destination. This paper shows that, with convex, strictly increasing and differentiable arc cost functions, an atomic splittable player always has an optimal unilateral decentralization strategy. Besides, unilateral decentralization gives her the same advantage as being the leader in a Stackelberg congestion game. Finally, unilateral decentralization of an atomic player has a negative impact on the social cost and on the costs of the other players at the equilibrium of the congestion game.
\end{abstract}
 \medskip
 
 \textbf{Keywords}~~ routing, decentralization, Stackelberg game, composite congestion game

\section{Introduction}
This paper introduces strategic decentralization into composite network congestion games, and studies its properties in a specific subclass of such games. A player decentralizes her decision-making if she lets each of her deputies decide independently how to send the part of her stock deputed to him from its origin to its destination. A unilateral decentralization can be beneficial or deleterious for the decentralizing player herself, and it also has an influence on the other players' utility and the social welfare as well. This paper provides a detailed analysis of these problems in the case where all the players have the same binary choice.

In a network congestion game, i.e. routing game, each player has a certain quantity of stock and a finite set of choices. A choice is a directed, acyclic path from the player's origin to her destination. A player with a stock of infinitesimal weight is \emph{nonatomic}. She has to attribute her stock to only one choice. A player with a stock of strictly positive weight is \emph{atomic}. She (more rigorously, her stock) is \emph{splittable} if she can divide it into several parts and affect each part to a different choice. She can also be non splittable, which is the case originally studied in the seminal work of Rosenthal 1973 \cite{Ro73a} on congestion games. This paper considers only the splittable case so that the word {\em splittable} is often omitted. A path is composed of a series of arcs, and the cost of a path is the sum of the costs of its component arcs. The cost entailed to a user of an arc depends on the total weight of the stocks on it as well as on the quantity of that user's stock on it. A player wishes to minimize her cost, which is the total cost to her stock. A game with both nonatomic and atomic players is called a {\em composite} game. An equilibrium in a composite congestion game is called a \emph{composite equilibrium} (CE for short) (Harker 1988 \cite{Ha88}, Boulogne et al. 2002 \cite{Bou02}, Yang and Zhang 2008 \cite{YangZh2008}, Wan 2012 \cite{Wan11}). An equilibrium does not necessarily minimize the social cost, i.e. the total cost to all the players.

In a composite congestion game, an atomic player of weight $m$ decentralizes if she is replaced by a composite set of players called her \emph{deputies} (i.e. $n$ atomic players of weight $\alpha^{1},\ldots,\alpha^{n}$ and a set of nonatomic players of total weight $\alpha^{0}$, such that $\sum^{n}_{i=0}\alpha^{i}=m$) who have the same choice set as her, and she collects the sum of her deputies' costs as her own. 

Here is an example of advantageous decentralization. Two atomic players both have a stock of weight $\frac{1}{2}$ to send from $O$ to $D$. Two parallel arcs link $O$ to $D$, with per-unit cost function $c_{1}(t)=t+10$ and $c_{2}(t)=10t+1$ respectively. At the equilibrium, both players send weight $\frac{2}{11}$ on the first arc and $\frac{7}{22}$ on the second one. The cost is $4.14$ to both players and the social cost is $8.28$. If player 1 deputes her stock to two atomic deputies both of weight $\frac{1}{4}$, then at the equilibrium of the resulting congestion game, both deputies send weight $\frac{1}{44}$ on the first arc, while player~$2$ sends weight $\frac{1}{4}$ there. The cost is $2.06$ to both deputies of player 1. Hence player 1 gains by decentralizing because her current cost $4.12$ is lower than $4.14$. However, player 2's cost is now $4.59$ and the social cost is $8.71$, both higher than before.

Assuming that the arc cost functions are convex, strictly increasing and continuously differentiable in congestion, this paper obtains the following properties of unilateral decentralization in composite congestion games with binary choice or, equivalently, in a two-terminal two-parallel-arc composite routing game:

(i) For the atomic player who decentralizes unilaterally, all her decentralization strategies are weakly dominated by single-atomic ones which depute her stock to at most one atomic deputies in addition to nonatomic ones (Theorem~\ref{thm:sym_single}). A fortiori, she possesses an optimal decentralization strategy (Theorem~\ref{thm:best}), which depends on her relative size among all the players. 

(ii) Unilateral decentralization gives an atomic player the same advantage as being the leader in a Stackelberg congestion game (Theorem~\ref{thm:stackelberg}).

(iii) After the unilateral decentralization of an atomic player, the social cost at the equilibrium increases or does not change, and the cost to each of her opponents increases or does not change (Theorem~\ref{thm:social_2}).

Although the above results are obtained in the specific setting of binary choice games, the goal of this paper is to introduce the notion of strategic decentralization into composite congestion games, to point out its significance, and to initiate a systematic study of its properties.  

The paper is organized as follows. Section~\ref{sec:chp3_sec2} presents the model, defines decentralization, and shows the special role of single-atomic decentralization strategies. Section~\ref{sec:deleg} proves the existence of an optimal unilateral decentralization strategy, and shows that unilateral decentralization gives an atomic player the same advantage as being the leader in a Stackelberg congestion game. Section~\ref{sec:chp3_sec4} focuses on the impact of unilateral decentralization on the social cost and the other players' cost. Section~\ref{sec:discussion} concludes. The proofs and auxiliary results are regrouped in Section~\ref{sec:proofs}.

\subsection*{Related literature}
The ``inverse'' concept of decentralization -- coalition formation or collusion between players -- has been extensively studied. Hayrapetyan et al. 2006 \cite{Hay06} first define the {\em price of collusion} (PoC) of a parallel network to be the ratio between the worst equilibrium social cost after the nonatomic players form disjoint coalitions and the worst equilibrium social cost without coalitions. Bhaskar et al. 2010 \cite{BhaAl09b} extended this study to series-parallel networks. (A series-parallel network can be constructed by merging in series or in parallel several graphs of parallel arcs.) This index is closely related to another important notion: the {\em price of anarchy} (PoA), which is introduced by Koutsoupias and Papadimitriou 1999 \cite{Kou99} (and \cite{Papa01}) as the ratio between the worst equilibrium social cost and the minimal social cost in nonatomic games. Cominetti et al. 2009 \cite{Com09} derives the first bounds on the PoA with atomic players. For a specific network structure, one can deduce the PoC by the PoA with atomic players and the PoA with nonatomic players. Further results on the bound of the PoA with atomic players are obtained in Harks 2011 \cite{Harks2011}, Roughgarden and Schoppmann 2011 \cite{Rough11} and Bhaskar et al. 2010 \cite{BhaAl09b}. Roughgarden and Tardos 2002 \cite{RouTar02} and Correa et al. 2008 \cite{CorAl08} provide fundamental results on the bound of the PoA with nonatomic players. PoA in nonatomic games with asymmetric costs or elastic demands are studied in \cite{Perakis2007} and \cite{ChauSim2003}, among others.

Beyond the coalitions formed by nonatomic players, Cominetti et al. 2009 \cite{Com09}, Altman et al. 2011 \cite{AltAl11b}, and Huang 2013 \cite{Huang2013} consider those formed by atomic players. Their results can be interpreted as the impact of certain kinds of collusion and hence, the ``inverse" of it, decentralization, on the social cost. Wan 2012 \cite{Wan11} studies the impact of coalition formation on the nonatomic players' cost outside the coalition in parallel-link networks. In terms of the impact of coalition formation on the cost of the coalition members themselves, Cominetti et al. 2009 \cite{Com09}, Altman et al. 2011 \cite{AltAl11b} and Wan 2012 \cite{Wan11} provide examples in different contexts of disadvantageous coalition formation for the members themselves. These are actually examples of advantageous decentralization. Finally, for works on strategic decentralization, one can cite Sorin and Wan 2013 \cite{SW12} in integer-splitting congestion games, and Baye et al. 1996 \cite{BCJ1996} in industrial organization (where they call the strategic decentralization of a firm ``divisionalization''). 

Finally, let us point out that the above-mentioned coalition formation is studied by the approach of comparative statics in a noncooperative game setting. It is different from the cooperative routing games studied in Quant et al. 2006 \cite{QBR2006} and Blocq and Orda \cite{BlocqOrda2014}.

\section{Model and preliminary results}\label{sec:chp3_sec2}
\subsection{Binary choice composite congestion games}\label{sect:notations}
\begin{figure}[htbp!]
\begin{center}
\begin{tikzpicture}
 \node[draw,circle,scale=0.8] (O)at(-2,0) {$O$};
 \node[draw,circle,scale=0.8] (D)at(2,0) {$D$};
\draw[->,>=latex] (O) to[bend left=25] node[midway,above,scale=0.8]{$c_{1}(t)$}(D);
\draw[->,>=latex] (O) to[bend right=25] node[midway,below,scale=0.8]{$c_{2}(t)$}(D);
\end{tikzpicture}
\end{center}
\caption{A binary choice congestion game.}
\end{figure}
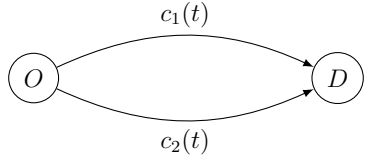

Nodes $O$ and $D$ are linked by two parallel arcs. The per-unit cost function of arc $r$ is $c_{r}$, for $r=1,2$. When the total weight of stocks on arc $r$ is $t$, the cost to each unit of them is $c_{r}(t)$. Both $c_{1}$ and $c_{2}$ are defined on $\Omega$, a neighbourhood of $[0,\bar{M}]$, with $\bar{M}>0$. They satisfy the following assumption throughout this paper.
\begin{assumption}\label{assp:base_chap3}
Both $c_{1}$ and $c_{2}$ are strictly increasing, convex and continuously differentiable on $\Omega$, and non-negative on $[0,\bar{M}]$.
\end{assumption}

There is a continuum of nonatomic players of total weight $T^{0}$, and $N$ atomic players of strictly positive weight $T^{1},T^{2},\ldots,T^{N}$ respectively, where $N\in \mathbb{N}$. If there are no nonatomic (resp. no atomic) players, then $T^{0}=0$ (resp. $N=0$). Let $I=\{0,1,\ldots,N\}$. The player profile is denoted by $T=(T^{i})_{i\in I}$, and their total weight is $M=\sum_{i\in I}T^{i}$, with $M<\bar{M}$.

The profile of the nonatomic players' strategies is described by their flow $x^{0}=(x^{0}_{1},x^{0}_{2})$, where $x^{0}_{r}$ is the total weight of the nonatomic players on arc $r$. The strategy of atomic player~$i$ is specified by her flow $x^{i}=(x^{i}_{1},x^{i}_{2})$, where $x^{i}_{r}$ is the weight that she sends by arc $r$. Call $x=(x^{i})_{i\in I}$ the (system) flow. Denote respectively by $X^{i}=\{x^{i}\in \mathbb{R}^{2}_+\,|\,x^{i}_{1}+x^{i}_{2}=T^{i}\}$ the space of feasible flows for the nonatomic players or an atomic player~$i$, and by $X=\prod_{i\in I}X^{i}$ the space of feasible system flows. Let $\xi=(\xi_r)_{r\in \{1,2\}}$ be a vector function defined on $X$ by $\xi_{r}(x)=\sum_{i\in I}x^{i}_{r}$, i.e. the aggregate weight on arc $r$. For $i\in I$, let $x^{-i}=(x^{j})_{j\in I  \!\setminus \! \{i\}}$.

With flow $x$, the cost to a nonatomic player taking arc $r$ is $c_{r}(\xi_{r}(x))$. The cost to atomic player~$i$ is $u^{i}(x)=x^{i}_{1}c_{1}(\xi_{1}(x))+x^{i}_{2}c_{2}(\xi_{2}(x))$. The social cost is $C\!S(x)=\xi_{1}(x)c_{1}(\xi_{1}(x))+\xi_{2}(x)c_{2}(\xi_{2}(x))$.
\smallskip

Let this composite congestion game be denoted by $\Gamma(T)$. Flow $x\in X$ is a {\em composite equilibrium} (CE) of $\Gamma(T)$ if (Harker 1988 \cite{Ha88}): 

(a) for $r\in \{1,2\}$, if $x^{0}_{r}>0$, then $c_{r}(\xi_{r}(x))\leq c_{s}(\xi_{s}(x))$ for all $s\in \{1,2\}$; and

(b) for $i\in I  \!\setminus  \!\{0\}$, $x^{i}$ minimizes $u^{i}(\,\cdot,\,x^{-i})$ on $X^{i}$.
\smallskip

Like all composite congestion games taking place in a two-terminal parallel-arc networks, game $\Gamma(T)$ always admits a unique CE. The reader is referred to  \cite{Rich07} or \cite{Wan11} for a proof. For the uniqueness of equilibria in congestion games with different types of players and in more general networks, see, for example, \cite{Milchtaich2005,Rich07,MeunierPradeau2014} and \cite{BLHH2015}. Let the nonatomic players' common cost at the unique CE $x$ be denoted by $u^0(x)$.

According to Assumption~\ref{assp:base_chap3}, there exists at most one number $\hat{\xi}\in [0,M]$ such that $c_{1}(\hat{\xi})=c_{2}(M-\hat{\xi})$, and $\xi$ exists if and only if $c_{1}(M)\geq c_{2}(0)$ and $c_{2}(M)\geq c_{1}(0)$.

There are four possible cases concerning the relations between $c_1(0)$, $c_1(M)$, $c_2(0)$, $c_2(M)$ and $\hat{\xi}$. Two of them are listed in the following assumption and studied in this paper. 
\begin{assumption}\label{assp:c1_less_c2}
One of the following two conditions holds:

\textsc{(i)} $c_{1}(M)<c_{2}(0)$.

\textsc{(ii)} $c_{1}(M)\geq c_{2}(0), c_{2}(M)\geq c_{1}(0)$, and
\begin{equation}\label{eq:assp_c1}
\hat{\xi}\,c'_{1}(\hat{\xi}) \geq  (M-\hat{\xi})\,c'_{2}(M-\hat{\xi}).
\end{equation}
\end{assumption}
The other two cases are symmetric to them:  
\textsc{(iii)} $c_{2}(M)<c_{1}(0)$; \textsc{(iv)} $c_{1}(M)\geq c_{2}(0)$, $c_{2}(M)\geq c_{1}(0)$ and $\hat{\xi}\,c'_{1}(\hat{\xi}) \leq  (M-\hat{\xi})\,c'_{2}(M-\hat{\xi})$. One can prove that cases \textsc{(i)} and \textsc{(ii)} correspond to the situation where arc 1 is less costly than arc 2 at the CE $x$ of any composite congestion game taking place in this network with the total weight of the players being $M$ (cf. Lemma~\ref{prop:c1_less_c2}), whereas cases \textsc{(iii)} and \textsc{(iv)} correspond to the inverse. 

\subsection{Decentralization strategies}
In composite congestion game $\Gamma(T)$, an atomic player~$l$ of weight $T^{l}$ {\em decentralizes} if she is replaced by a finite number $n\in \mathbb{N}$ of atomic players of weight $\alpha^{1},\alpha^{2},\ldots,\alpha^{n}$ (in a non-increasing order) and a continuum of nonatomic players of total weight $\alpha^{0}$ such that $\sum^{n}_{i=1}\!\alpha^{i}+\alpha^{0}\!=\!T^{l}$. These nonatomic players of total weight $\alpha^{0}$ and the $n$ atomic players are called her {\em deputies}. There can be only atomic or only nonatomic deputies.

A ({\em decentralization}) {\em strategy} of atomic player~$l$ is a profile $\alpha\!=\!(\alpha^{0},\alpha^{1},\ldots,\alpha^{n})$ of her deputies. 
Atomic player~$l$'s strategy space is denoted by $\cS^{l} = \bigcup^{+\infty}_{n=0}\cS^{l,n}$, where $\cS^{l,n}$ is the set of strategies designating $n$ atomic deputies:
\begin{equation*}
 \cS^{l,n} \!= \! \{\alpha\!=\!(\alpha^{0},\alpha^{1},\ldots,\alpha^{n})\!\in\! \mathbb{R}^{n+1}_+\, |\,\alpha^{1}\!\geq \!\cdots \!\geq\! \alpha^{n}\!>\!0;\; \sum^{n}_{i=0}\!\alpha^{i}\!=\!T^{l}\}.
\end{equation*}

There are some specific classes of strategies worth noticing. 

\textbullet~~$\alpha$ is the {\em nonatomic strategy}, denoted by $\underline{\alpha}$, if $\alpha^{0}=T^{l}$ or, equivalently, $n=0$. It is the unique element in $\cS^{l,0}$. 

\textbullet~~$\alpha$ is the {\em trivial strategy}, denoted by $\bar{\alpha}$, if $n=1$ and $\alpha^{1}=T^{l}$. When playing this strategy, atomic player~$l$ does not decentralize. 

\textbullet~~$\alpha$ is a {\em single-atomic strategy} (SA strategy for short), if $\alpha=\underline{\alpha}$, or $n=1$ and $\alpha^{1}\in (0,\,T^{l}]$. The set of SA strategies is $\cS^{l,0}\cup \cS^{l,1}$. An SA strategy is determined, and from now on also denoted, by the weight of the unique atomic deputy (if there is one) and 0 if there is none. The space of SA strategies is thus isometric to the closed interval $S\!A^{l}:=[0,\,T^{l}]$. The nonatomic strategy and the trivial strategy are both SA.

\subsection{Unilateral decentralization and SA strategies}\label{sect:five}
This paper focuses on unilateral decentralization. Suppose that atomic player~$N$ decentralizes unilaterally. Let $T^{-N}=\{T^{j}\}_{j=0}^{N-1}$ be the profile of atomic player~$N$'s opponents. 

After player~$N$ decentralizes according to strategy $\alpha\in \mathcal{S}^N$, the profile of the players in the network is denoted by $(\alpha,T^{-N})$. Rigorously, the player profile is $(\alpha^0+T^0,\alpha^1,\ldots,\alpha^n,T^1,\ldots,T^{N-1})$. Instead of playing a congestion game against her opponents $T^{-N}$ by herself, atomic player~$N$ let her deputies do it, and autonomously, i.e. without any cooperation between them. This congestion game $\Gamma(\alpha,T^{-N})$ is hence induced by $N$'s decentralization strategy $\alpha$. Denote its unique CE by $z_{\alpha}=((x^{i}(z_{\alpha}))^{n}_{i=0},\,(y^{j}(z_{\alpha}))^{N-1}_{j=0})$, with $x^{i}(z_{\alpha})=(x^{i}_{1}(z_{\alpha}),x^{i}_{2}(z_{\alpha}))$ and $y^{j}(z_{\alpha})=(y^{j}_{1}(z_{\alpha}),y^{j}_{2}(z_{\alpha}))$. Atomic deputy of weight $\alpha^{i}$ sends weight $x^{i}_{r}(z_{\alpha})$ on arc $r$, and nonatomic deputies of total weight $x^{0}_{r}(z_{\alpha})$ among $\alpha^{0}$ choose arc $r$. Let $x_{r}(z_{\alpha})=\sum^{n}_{i=0}x^{i}_{r}$ be the total weight of atomic player~$N$'s stock on arc $r$, whereas $y_{r}(z_{\alpha})=\sum^{N-1}_{j=0}y^j_r$ be the aggregate weight of her opponents' stock there. Still denote the aggregate weight on arc $r$ by $\xi_{r}(z_\alpha)$. 

Assume that the CE $z_\alpha$ is always attained in congestion game $\Gamma(\alpha,T^{-N})$. Player~$N$'s cost for playing $\alpha$ is defined as the total cost to her deputies at $z_\alpha$:
\begin{equation}\label{def:cost_decentralization}
U^N(\alpha,T^{-N})=x_{1}(z_\alpha)\,c_{1}(\xi_{1}(z_\alpha))+x_{2}(z_\alpha)\,c_{2}(\xi_{2}(z_\alpha)).
\end{equation}


Two decentralization strategies $\alpha$ and $\tilde{\alpha}$ in $\cS^{N}$ are {\em equivalent with respect to} the opponents' profile $T^{-N}$ if they yield the same cost to atomic player~$N$: 
\begin{equation*}
U^N(\alpha,T^{-N})=U^N(\tilde{\alpha},T^{-N}). 
\end{equation*}
In this paper, the equivalency between two decentralization strategies of $N$ is always with respect to $T^{-N}$, hence it is no longer specified.

SA strategies play a special role among all the decentralization strategies.
\begin{theorem}\label{thm:sym_single}
For any strategy $\alpha\in \cS^{N}$, there is an SA strategy $s$ that is equivalent to $\alpha$. 
\end{theorem}
\begin{proof}
Lemmas~\ref{thm:mode_1_cond}--\ref{thm:mode_5_cond} give the subset of SA strategies equivalent to $\alpha$. 
\end{proof}


\section{Optimal decentralization}\label{sec:deleg}
\subsection{Optimal decentralization strategy}\label{chp4_sec4}
This section investigates in the existence of an optimal decentralization strategy $\alpha$ of atomic player~$N$, i.e. the one minimizing $U^N(\alpha,T^{-N})$. Note that the space of decentralization strategies of $N$, $\mathcal{S}^N=\bigcup_{n=0}^{+\infty}\mathcal{S}^{N,n}$, is the union of a countably infinite number of simplexes. Hence, the existence of an optimal strategy should not be taken for granted.  
 
According to Theorem~\ref{thm:sym_single}, if $N$ has an optimal decentralization strategy, there must be an equivalent SA strategy of her. Therefore, instead of studying the existence of an optimal strategy in $\mathcal{S}^N$, one need only focus on the space of SA strategies, $S\!A^N=[0,T^N]$, which is a compact set in $\mathbb{R}$. An optimal strategy thus exists if player~$N$'s cost is continuous in her SA strategy on $S\!A^N$.

As the following theorem shows, an SA optimal strategy does exist and its explicit form depends on the relative weight of atomic player~$N$ among all the players. In particular, when $M$ and $T^1$, ..., $T^{N-1}$ are fixed, three cases can be distinguished. They are respectively called \emph{nonatomic}, \emph{trivial} and \emph{nontrivial}, according to the form of the  SA optimal strategies.

%
%

\begin{theorem}\label{thm:best}
Atomic player~$N$ has an optimal decentralization strategy which minimizes her cost at the CE of the induced composite congestion game. 
More precisely, denoting $H=\frac{c_{2}(0)-c_{1}(M)}{c'_{1}(M)}$, one has:

1 {\em (nonatomic)}. Every strategy of atomic player~$N$ is optimal and is equivalent to the nonatomic strategy $\underline{\alpha}$, if one of the two following holds:

(1.1) $H>0$ and either (i) $T^{i}\leq H$ for all $i\in I\setminus \{0\}$, or (ii) $N>1$, $\max_{1\leq i \leq N-1}{T^i}> H$, $T^{N}\leq C_0$.

(1.2) $H\leq 0$ and either (i) $\sum_{i=1}^{N}T^i\leq C_1$, or (ii) $N>1$, $\sum_{i=1}^{N-1}T^i> C_1$, and $T^{N}\leq C_2$.

2 {\em (trivial)}. Atomic player~$N$'s unique optimal strategy is the trivial one $\bar{\alpha}$, i.e. not decentralizing, if one of the two following holds:

(2.1) $H>0$, $T^{N}> H$,  and $\max_{1\leq i \leq N-1}{T^i}\leq H$ in the case that $N>1$.

(2.2) $H \leq 0$, $N=1$ and $T^{N}> C_1$.

3 {\em (nontrivial)}. Atomic player~$N$ has at least one optimal strategy which is not necessarily nonatomic or trivial in the remaining situations.
%
%

Here $C_0$, $C_1$ and $C_2$ are strictly positive constants, determined by $M$ and $T^1$, \ldots, $T^{N-1}$.
\end{theorem}
Some remarks on the profitability of strategic decentralization of atomic player~$N$ are necessary here. It is known (Orda et al. 1993 \cite{Orda93}, Wan 2012 \cite{Wan11}, also see Lemma~\ref{prop:c1_less_c2}) that smaller players are more likely to \emph{free ride} by using the least costly arc(s). The bigger a player is, the more she tends to use the more expensive arc to internalize the negative externality of her choice on her own stock. Accordingly, as Theorem~\ref{thm:best} shows, the optimal choice of atomic player~$N$ also depends on her relative size among all the players. In the nonatomic case, either atomic player~$N$ is too small compared with her atomic opponents, or all the players are small. Then, she can never change the outcome of the congestion game by unilateral decentralization, because she and her potential deputies always behave like free-riders. In the trivial case, atomic player~$N$ is very big compared to her opponents. If she decentralizes, she cannot free ride on them to her advantage because they are too small. Moreover, since her deputies are not internalizing enough their externalities, she actually loses by decentralizing. In the non trivial case, atomic player~$N$ is neither too small to be always a free-rider, nor big enough to be dominating. In this case, by some appropriate decentralization, she manages to free ride on her atomic opponents.

Besides, let us point out that although the paper studies a one-person decision problem where atomic player~$N$ chooses how to decentralize, an alternative formulation of the problem as an extensive form game is also possible. By the definition of $N$'s cost $U^N(\alpha,T^{-N})$ induced by her choice $\alpha$, the problem can be considered as a two-stage game called \emph{unilateral decentralization game}. In the first stage, only $N$ makes a move by choosing a decentralization strategy. In the second stage, her deputies created by this choice as well as the players in $T^{-N}$ play a composite congestion game. Then, an optimal decentralization strategy $\alpha$ of atomic player~$N$ and the corresponding CE $z_\alpha$ in $\Gamma(\alpha,T^{-N})$ constitute a subgame perfect Nash equilibrium (SPNE) of the unilateral decentralization game. Theorem~\ref{thm:best} shows that an SPNE exists in this game. 

However, this unilateral decentralization game is not precisely a Stackelberg game in the sense that the leader, atomic player~$N$, moves in the first stage and the followers, $T^{-N}$, move in the second stage. As a matter of fact, by decentralizing unilaterally in the first stage, player~$N$ creates new players, i.e. her deputies, who will participate in the second stage. Hence, while making her choice in the first stage, player~$N$ should anticipate not only the action of her opponents $T^{-N}$ but also that of her own deputies in the second stage.

\subsection{Unilateral decentralization and Stackelberg game}
This subsection compares the advantage of being able to decentralize unilaterally and that of being the leader in a Stackelberg congestion game. Stackelberg-type behaviour in routing is studied in \cite{KLO1997,Roughgarden2004,YZM2007,KK2009} and \cite{BHS2010}, etc.

Let $S\!\Gamma(T^N, T^{-N})$ be the Stackelberg composite congestion game where atomic player~$N$ is the leader and the players in  $T^{-N}$ are the followers. In the first stage of the game, player $N$ chooses how to distribute her flow on the two arcs. Then in the second stage, the followers choose how to distribute their flows. Let player~$N$'s strategy in the first stage be denoted by $x=(x_1, x_2)\in X^N$. Given $x$, the followers in $T^{-N}$ play a composite congestion game denoted by $\Gamma_x(T^{-N})$ whose unique CE is denoted by $Z_x$. Let $\Pi^N(x,T^{-N})$ denote player $N$'s cost at $(x,Z_x)$. Then a subgame perfect Nash equilibrium (SPNE) of the game is attained, if it exits, at $(x^*,Z_{x^*})$ where $x^*=\arg\min _{x\in X^N} \Pi^N(x,T^{-N})$.

At the end of the previous subsection, the optimal unilateral decentralization problem is formulated as a two-stage game, where atomic player~$N$ is also the first mover.  Although that game is not exactly a Stackelberg game as $S\!\Gamma(T^N, T^{-N})$, atomic player~$N$ has the same advantage as the first mover in both games, as the following theorem shows.

\begin{theorem}\label{thm:stackelberg}
Stackelberg game $S\!\Gamma(T^N, T^{-N})$ admits an SPNE. 

Besides, by playing an optimal decentralization strategy $\alpha\in \cS^N$, atomic player~$N$ has the same cost $U^N(\alpha,T^{-N})$ as her SPNE cost in Stackelberg game $S\!\Gamma(T^N, T^{-N})$.
\end{theorem}

Obviously, player~$N$ cannot do worse in the Stackelberg game than in the unilateral decentralization game, because she can always do by herself what she anticipates her deputies to do. 
On the contrary, there can be a strategy $(x_1,x_2)$ of player~$N$ in the Stackelberg game that cannot be ``mimicked'' by a decentralization strategy. In other words, no decentralization strategy $\alpha$ satisfies that, at the CE of the induced congestion game, the aggregate weight of player $N$'s deputies' stock on arc $r$ is exactly $x_r$. However, such a strategy $(x_1,x_2)$ cannot be optimal as the proof of the theorem shows. 

One should not deduce from Theorem~\ref{thm:stackelberg} that the study of unilateral decentralization is useless because, once being the first mover, an atomic player need only do what her optimal deputies would have done. As pointed out in Sorin and Wan 2013 \cite{SW12}, in a congestion game, a player's choice has an influence on other players' costs not via her identity (i.e. anonymously) but via the weight of her stock attributing to each particular choice. The behavior of decentralization is thus feasible in terms of strategies and undetectable to the others. These two features distinguish decentralization, as a strategic opportunity, from the first mover's advantage in a Stackelberg congestion game.

\section{Impact of strategic decentralization}\label{sec:chp3_sec4}
This section studies the impact of atomic player~$N$'s  unilateral decentralization on the other players' costs and the social cost. Recall that the trivial decentralization strategy $\bar{\alpha}$ corresponds to not decentralizing. Thus one has only to compare the cost to the players in $T^{-N}$ and the social cost at $z_{\bar{\alpha}}$ with those at $z_\alpha$, for an arbitrary strategy $\alpha \in \cS^N$.

\begin{theorem}\label{thm:social_2}
Suppose that atomic player~$N$ decentralizes according to strategy $\alpha\in \cS^N$. Then at the CE of the induced congestion game, the social cost and the cost to each player in $T^{-N}$ are not lower than at the CE of the congestion game without decentralization:
\begin{equation}\label{eq:neg_ext}
 CS(z_\alpha)\geq CS(z_{\bar{\alpha}})\; \text{ and }\; u^{j}(z_\alpha)\geq u^{j}(z_{\bar{\alpha}}), \;\,\forall j\in I \!\setminus\! \{N\}.
\end{equation}
In addition, equalities hold only in the nonatomic case, i.e. when the assumptions in the first case of Theorem~\ref{thm:best} hold. 
\end{theorem}

Recall that the smaller a player is, the more she tends to free ride by using the less costly arc (cf. remark after Theorem~\ref{thm:best}) and exerts a higher externality on the other players. By decentralizing, atomic player~$N$ lets her deputies, who are smaller than her, free ride more aggressively. In other words, her deputies put more weight on the less costly arc, arc 1, than she herself would have done. Not only do these deputies of $N$ increase the cost of the less expensive arc for its users, they also drive the others to put more weight on the more expensive arc, arc 2. In this way, the unilateral decentralization of $N$ increases the social cost as well as her opponents' costs. 

Meunier and Pradeau \cite[Theorem 3]{MeunierPradeau2015} recover a particular case of Theorem~\ref{thm:social_2} concerning social cost and without nonatomic players. They show that in an atomic game taking place in a two-terminal two parallel-arc network with increasing, strictly convex and differentiable arc cost functions, if an atomic player~$i$ transfers some of her stock to a smaller atomic player~$j$ ($T^j\leq T^i$), then the social cost at the CE increases or remain constant after the transfer. In particular, when taking $T^j=0$, this transfer of stock is equivalent to the decentralization of atomic player~$i$ who deputes her stock to two atomic deputies. 

\section{Discussion and perspectives}\label{sec:discussion}
This paper first introduces strategic decentralization behavior into composite congestion games. Then, in the particular setting of binary choice case, it shows (i) the existence of optimal unilateral decentralization which depends on the relative size of the decentralizing player, (ii) the ``equivalence'' between being the only player to decentralize and being the leader in a Stackelberg game, and (iii) the negative impacts of unilateral decentralization on the others.  

With more general network structure, Theorem~\ref{thm:social_2} is no longer valid.

In the setting of single OD (where all the players share the same origin and the same destination), Huang 2013 \cite{Huang2013} provides the following example adapted in our language. In the network shown on the left hand side of Figure~\ref{fig:cominetti}, the social cost at the unique equilibrium with two atomic players of weight $200$ and $21$ respectively increases when the atomic player of weight $21$ deputes her stock to two atomic deputies of weight $20.9$ and $0.1$ respectively. 

 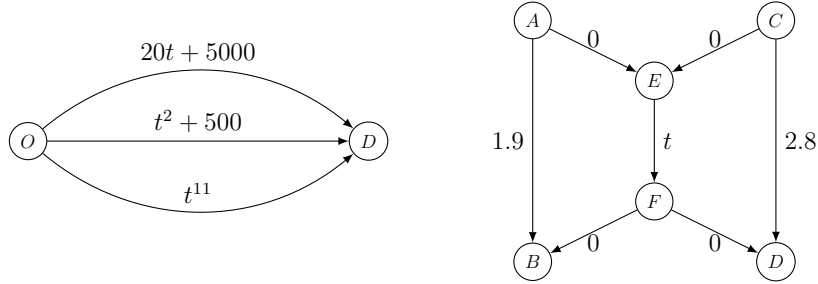
\begin{figure}[htbp!]
\begin{center}
 \begin{tikzpicture}
\begin{scope}[scale=1.6]
 \node[draw,circle,scale=0.6] (O)at(-1.4,0) {$O$};
 \node[draw,circle,scale=0.6] (D)at(1.4,0) {$D$};
\draw[->,>=latex] (O) to[bend left=40] node[midway,above,scale=0.8]{$20t+5000$}(D);
\draw[->,>=latex] (O) to node[midway,above,scale=0.8]{$t^{2}+500$}(D);
\draw[->,>=latex] (O) to[bend right=40] node[midway,above,scale=0.8]{$t^{11}$}(D);
\end{scope}
\begin{scope}[shift={(6,0)}]
 \node[draw,circle,scale=0.6] (B)at(-1.6,-1.6) {$B$};
 \node[draw,circle,scale=0.6] (A)at(-1.6,1.6) {$A$};
 \node[draw,circle,scale=0.6] (C)at(1.6,1.6) {$C$};
 \node[draw,circle,scale=0.6] (D)at(1.6,-1.6) {$D$};
 \node[draw,circle,scale=0.6] (E)at(0,0.8) {$E$};
 \node[draw,circle,scale=0.6] (F)at(0,-0.8) {$F$};
 \draw[->,>=latex] (A) to node[midway,left,scale=0.8]{$1.9$}(B);
 \draw[->,>=latex] (A) to node[]{}(E);
 \node[scale=0.8] (A1)at(-0.8,1.35) {$0$};
 \node[scale=0.8] (C1)at(0.8,1.35) {$0$};
 \node[scale=0.8] (B1)at(-0.8,-1.35) {$0$};
 \node[scale=0.8] (D1)at(0.8,-1.35) {$0$};
 \draw[->,>=latex] (C) to node[midway,right,scale=0.8]{$2.8$}(D);
 \draw[->,>=latex] (C) to node[]{}(E);
 \draw[->,>=latex] (E) to node[midway,right,scale=0.8]{$t$}(F);
 \draw[->,>=latex] (F) to node[]{}(B);
 \draw[->,>=latex] (F) to node[]{}(D);
\end{scope}
 \end{tikzpicture}
\caption{Counter-examples.\label{fig:cominetti}}
\end{center}
 \end{figure}

Cominetti et al. 2009 \cite{Com09} provide the following example in a multiple OD game. 
The network is shown on the right hand side of Figure~\ref{fig:cominetti}. Only arc $E\!-\!F$ has non constant cost function. A group of nonatomic players and an atomic player both have weight 1. The nonatomic group  has OD pair $(A,B)$ and available paths $A\!-\!B$ and $A\!-\!E\!-\!F\!-\!B$, while the atomic player has OD pair $(C,D)$ and available paths $C\!-\!D$ and $C\!-\!E\!-\!F\!-\!D$. The social cost at the CE is $3.89$. If the atomic player deputes her stock to a group of nonatomic players, the social cost at the new CE is $3.8$, lower than $3.89$. Altman et al. 2011 \cite{AltAl11b} provide an example of multiple OD game, where the decentralization of an atomic player to several atomic deputies increases the social cost. 

These examples show that our result concerning the social cost can neither be extended to single OD games with more than two parallel arcs, nor to multiple OD games even if each OD pair is linked by two parallel arcs. The validity of the result concerning the opponents of the decentralizaer remains an open question.

Besides, this paper only considers unilateral decentralization. Decentralization games in which all the atomic players decentralize simultaneously are worth examining. Another potential extension follows Sorin and Wan 2013 \cite{SW12}, where a deputy can also decentralize, and his deputies as well, and so on. In their case, a player has a finite integer weight and can only have deputies of integer weight. Therefore, sequential decentralization will terminate after a finite number of steps. In the setting of this paper, however, an atomic deputy of any size is able to decentralize. A possible approach consists in studying the asymptotic behavior of a sequential decentralization process. A first attempt to study one-shot simultaneous decentralization games as well as sequential decentralizing processes is made in Wan 2012 \cite{Wan12c}, with two atomic players and affine arc cost functions. 

\section{Auxiliary results and the proofs}\label{sec:proofs}
\subsection{Auxiliary functions, notations and properties of a CE}
Without loss of generality, assume from now on $T^1 \geq T^2 \geq \cdots \geq T^{N-1}$. Denote $M^{-N}=\sum_{i=0}^{N-1}T^i$ the total weight of atomic player~$N$'s opponents. Denote by $T^{[l]}=\sum_{j=1}^{l}T^{j}$ the total weight of the $l$ largest atomic players in $T^{-N}$, and let $T^{[0]}=0$.

Fix $\epsilon>0$. Functions $h$, $a$, and $F_{n}$ for $n\in \mathbb{N}$ are defined on $\Omega$ as follows:

\begin{align*}
& h(t)=
\begin{dcases}
\, \frac{c_{2}(M-t)-c_{1}(t)}{c'_{1}(t)}, \,&\text{if }\, 0\leq  t \leq M;\\
\, h(0)-\epsilon t, &\text{if }\,  t < 0;\\
\, h(M) -\epsilon (t-M), &\text{if }\, t > M,
\end{dcases}
\quad a(t)=
\begin{dcases}
\, \frac{c'_{2}(M-t)}{c'_{1}(t)}, \, &\text{if }\, 0\leq  t \leq M;\\
\, a(0),&\text{if }\,  t < 0;\\
\, a(M), &\text{if }\, t > M,
\end{dcases}
\end{align*}
\begin{equation*}
 F_{n}(t) = (M-t)\,(1+a(t))+n\,h(t),\qquad \, n\in \mathbb{N}.
\end{equation*}
\begin{equation*}
H \triangleq h(M),\quad A \triangleq a(\hat{\xi}).
\end{equation*}

The following facts are derived from Assumption~\ref{assp:base_chap3} and the above definition.

(i) $a$ is non-increasing, strictly positive and continuous on $\Omega$. 

(ii) $h$ and all $F_{n}$'s are strictly decreasing and continuous on $\Omega$.  Their inverse functions $h^{-1}$ and $F_{n}^{-1}$ are well defined on neighborhoods of $[h(M),h(0)]$ and $[F^{n}(M),F^{n}(0)]$ respectively, and are also strictly decreasing and continuous.

(iii) Case \textsc{(i)} in Assumption~\ref{assp:c1_less_c2} corresponds to $H>0$, whereas case \textsc{(ii)} corresponds to  $H\leq 0, h(0)\geq 0$ in addition to \eqref{eq:assp_c1}.

(iv) $h(\hat{\xi}) = 0$ and $F_{k}(\hat{\xi}) = (M-\hat{\xi})\,(1+A ) = F_{0}(\hat{\xi})$ for all $k$.

(v) $\hat{\xi}$ exists if and only if
\begin{equation}\label{prop:xi}
H\leq 0, \quad h(0)\geq 0.
\end{equation}

Recall a {\em necessary and sufficient} condition for $x\in X$ to be the CE of composite congestion game $\Gamma(T)$ \cite{Ha88}: for all $r\in\{1,2\}$,
\begin{align}
& x^{0}_{r}\!>\!0 \!\Rightarrow\! c_{r}(\xi_{r}(x))\!=\! \min_{s\!\in\!\{\!1,\!2\!\}}\!c_{s}(\xi_{s}(x)),\label{cond_ward}\\
& x^{i}_{r}\!>\!0 \!\Rightarrow\! c_{r}(\xi_{r}(x))\!+\!x^{i}_{r} c'_r(\xi_{r}(x))\!=\! \min_{s\!\in\!\{\!1,\!2\!\}}\!c_{s}(\xi_{s}(x))\!+\!x^{i}_{s} c'_s(\xi_{s}(x)),\;  \forall i \!\in\! I  \!\setminus \! \{0\}. \label{cond_atom}
\end{align}
%
The following lemma regroups some important properties of the CE of an arbitrary composite congestion game with aggregate stock weight $M$.

\begin{lemma}\label{prop:c1_less_c2}
At the CE $x$ of  $\Gamma(T)$, denote $\xi_r(x)$ simply by $\xi_r$.

1. For each arc $r\in \{1,2\}$, 

(1) if $x^{0}_{r}>0$, then $x^{i}_{r}>0$ for all atomic player~$i$; 

(2) if $T^{i}\geq T^{j}$ for atomic players $i$ and $j$, then $x^{i}_{r} \geq x^{j}_{r}$, and equality holds if and only if $T^{i}= T^{j}$ or $x^{i}_{r} = x^{j}_{r} = 0$; 

(3) for all atomic player~$i$ and for $s\neq r$, if $x^{i}_{r}>0$ and $x^{i}_{s}=0$, then $c_{r}(\xi_{r})<c_{s}(\xi_{s})$.

2. $c_{1}(\xi_{1})\leq c_{2}(\xi_{2})$.

3. If $H\leq 0$ (i.e. $c_{2}(0)\leq c_{1}(M)$), then $\xi_{1} \leq \hat{\xi}$.

4. $h(\xi_{1}) \geq 0$, and equality holds if and only if $H\leq 0$ and $\xi_{1}=\hat{\xi}$.
\end{lemma}
\begin{proof}
1. cf. Wan 2012 \cite{Wan11}.

2. In the case $H>0$ ($c_{1}(M)< c_{2}(0)$), $c_{1}(\xi_{1})\leq c_{2}(\xi_{2})$ is always true. In the case $H\leq 0$, suppose $c_{1}(\xi_{1})> c_{2}(\xi_{2})$. If $\xi_{1}\leq \hat{\xi}$ and thus $M-\xi_{1}\geq M-\hat{\xi}$, by the definition of $\hat{\xi}$, $c_{1}(\xi_{1}) \leq c_{1}(\hat{\xi}) = c_{2}(M-\hat{\xi}) \leq c_{2}(\xi_{2})$, contradicting the hypothesis that $c_{1}(\xi_{1})> c_{2}(\xi_{2})$. Therefore, $\xi_{1}> \hat{\xi}\geq 0$ and, consequently, $\xi_{2}<M-\hat{\xi}$. According to Assumptions~\ref{assp:base_chap3} and \ref{assp:c1_less_c2}, $\xi_{1}\,c'_{1}(\xi_{1})>\hat{\xi}\,c'_{1}(\hat{\xi}) \geq  (M-\hat{\xi})\,c'_{2}(M-\hat{\xi}) >  \xi_{2}\,c'_{2}(\xi_{2})$. Thus, $\xi_{1}\,c'_{1}(\xi_{1}) >  \xi_{2}\,c'_{2}(\xi_{2})$.

Notice that $N\!\geq\! 1$ because otherwise $T^{0}\!=\!M$, i.e. all the players are nonatomic and take the less expensive arc 2. Hence $c_{1}(0)\!>\! c_{2}(M)$, contradicting Assumption~\ref{assp:c1_less_c2}. It follows from the hypothesis $c_{1}(\xi_{1})\!>\! c_{2}(\xi_{2})$ and the first result of this lemma that there exists some $l\!\in\! \{1,\ldots, N\}$ such that $x^{i}_{1}\!>\!0,x^{i}_{2}\!>\!0$ for $ 1\!\leq\! i \!\leq\! l$ and, if $l\!<\!N$, $x^{i}_{1}=0,x^{i}_{2}=T^{i}$ for $l\!+\!1 \!\leq\! i \!\leq\! N$. According to eq.~\eqref{cond_atom},
\begin{equation}\label{eq:ch3_atomic}
c_{1}(\xi_{1})+x^{i}_{1}\,c'_{1}(\xi_{1}) = c_{2}(\xi_{2})+x^{i}_{2}\,c'_{2}(\xi_{2}),\hspace{1.14cm} \forall 1 \leq i \leq l.
\end{equation}
Besides, $x^{0}_{2}=T^{0}$ if $T^{0}>0$ because of eq.\eqref{cond_ward}. Summing eq.\eqref{eq:ch3_atomic} leads to $l c_{1}(\xi_{1})+\xi_{1} c'_{1}(\xi_{1}) =  l c_{2}(\xi_{2})+ (\xi_{2}-\sum^{N}_{i=l+1}T^{i}-T^{0}) c'_{2}(\xi_{2}) \leq l c_{2}(\xi_{2})+\xi_{2} c'_{1}(\xi_{1})$, and equality holds if and only if $T^{0}=0$ and $l=N$. But this is impossible because, by hypothesis, $c_{1}(\xi_{1}) > c_{2}(\xi_{2})$, and $\xi_{1}c'_{1}(\xi_{1}) >  \xi_{2}c'_{2}(\xi_{2})$.

Thus, $c_{1}(\xi_{1})\leq c_{2}(\xi_{2})$.

3. If $c_{2}(0)\leq c_{1}(M)$ and $\xi_{1} > \hat{\xi}$, then by Assumption~\ref{assp:base_chap3}, $c_{1}(\xi_{1}) > c_{1}(\hat{\xi}) = c_{2}(M-\hat{\xi}) > c_{2}(\xi_{2})$, which contradicts the fact that $c_{1}(\xi_{1})\leq c_{2}(\xi_{2})$.

4. If $c_{1}(M)< c_{2}(0)$ or, equivalently, $H>0$, the fact that $\xi_{1}\leq M$ implies that $h(\xi_{1})\geq h(M)=H>0$. If $c_{2}(0)\leq c_{1}(M)$, the fact that $\xi_{1} \leq \hat{\xi}$ implies that $h(\xi_{1})\geq h(\hat{\xi})=0$, and the equality holds if and only if $\xi_{1}=\hat{\xi}$.
\end{proof}


\subsection{Definition and properies of the four modes of $z_\alpha$} 
For a decentralization strategy $\alpha$, denote $\alpha^{[k]}=\sum_{i=1}^{k}\alpha^{i}$.  Also, for the sake of simplicity, $z_\alpha$ is often replaced by $z$, $x_r(z_\alpha)$ by $x_r$,  $y^j_r(z_\alpha)$ by $y^j_r$, $y_r(z_\alpha)$ by $y_r$, $\xi_r(z_\alpha)$ by $\xi_r$, if no confusion can arise. 

Define four modes of $z_\alpha$, the CE of congestion game $\Gamma(\alpha,T^{-N})$ as follows. They respectively correspond to: 1) All players put all their weight on arc 1; 2) Some atomic deputies of $N$ put some weight on arc 2; 3)  Some atomic players among $T^{-N}$ put some weight on arc 2, while all the deputies of $N$ put all their weight on arc 1; 4) All the atomic players put weight both on arc 1 and arc 2, while there are nonatomic players on both arcs.

The CE $z_\alpha$ is of {\em mode~1} if
\begin{equation}\label{mode_1}
\begin{cases}
\,c_{1}(\xi_{1})<c_{2}(\xi_{2});\\
\,x^{i}_{1}=\alpha^{i}, \;0\leq i \leq n;\quad y^{j}_{1}=T^{j}, \; 0\leq j \leq N-1.
\end{cases}
\end{equation}
\begin{table}[htbp!]
\begin{center}
{\scriptsize
\begin{tabular}{c c c c c c c c c}
 &$\alpha^{1}$&$\cdots$&$\alpha^{n}$&$\alpha^{0}$&$T^{1}$&$\cdots$&$T^{N-1}$&$T^{0}$\\
\hline
\rule{0pt}{2.5ex}{\em arc 1}&$\alpha^{1}$&$\cdots$&$\alpha^{n}$&$\alpha^{0}$&$T^{1}$&$\cdots$&$T^{N-1}$&$T^{0}$\\
{\em arc 2}&$0 $&$\cdots$&$ 0$&$0$&$0$&$ \cdots$&$ 0$&$0$\\
\end{tabular}
}\caption{Mode 1.}
\end{center}
\end{table}

The CE $z_\alpha$ is of {\em mode~2 specified by} $k\in \mathbb{N}^*$ and $l\in \{0,1,\ldots, N-1\}$, if
\begin{equation}\label{mode_2}
\begin{cases}
\, c_{1}(\xi_{1}) < c_{2}(\xi_{2}); \\
\, 1\!\leq\! k\!\leq\! n; \quad  x^{i}_{2}\!>\!0,\;  1 \leq i \leq  k; \;\\
\,  x^{i}_{2}=0,\;   i\!\in\! \!J_1 \triangleq\!\{k\!+\!1,\ldots, n\}\!\cup\! \{0\}, \text{ if } J_1\!\neq \!\emptyset;\\
\, y^{j}_{2}\!>\!0,\;   1\! \leq\! j\! \leq \!l, \,\text{ if } 1\!\leq \! l \!\leq \!N\!-\!1; \;\\
\, y^{j}_{2}\!=\!0,\; j\!\in\! \!J_2 \triangleq\!\{l\!+\!1,\ldots, N\!-\!1\}\!\cup\! \{0\}, \text{ if } J_2\!\neq \!\emptyset.
\end{cases}
\end{equation}
\begin{table}[htbp!]
\begin{center}
{\scriptsize
\begin{tabular}{c c c c c c c c c c c c c c c}
 &$\alpha^{1}$&$\cdots$&$\alpha^{k}$&$\alpha^{k+1}$&$\cdots$&$\alpha^{n}$&$\alpha^{0}$&$T^{1}$&$\cdots$&$T^{l}$&$T^{l+1}$&$\cdots$&$T^{N-1}$&$T^{0}$\\
\hline
\rule{0pt}{2.5ex}{\em arc 1}&$x^{1}_{1}$&$\cdots$&$x^{k}_{1}$&$\alpha^{k+1}$&$\cdots$&$\alpha^{n}$&$\alpha^{0}$&$y^{1}_{1}$&$\cdots$&$y^{l}_{1}$&$T^{l+1}$&$\cdots$&$T^{N-1}$&$T^{0}$\\
{\em arc 2}&$x^{1}_{2}$&$\cdots$&$x^{k}_{2}$&$0$&$\cdots$&$0$&$0$&$y^{1}_{2}$&$\cdots$&$y^{l}_{2}$&$0$&$ \cdots$&$ 0$&$0$
\end{tabular}
}
\caption{Mode 2.}
\end{center}
\end{table}

The CE $z_\alpha$ is of {\em mode~3 specified by} $l\in \{1,\ldots, N-1\}$, if
\begin{equation}\label{mode_3}
\begin{cases}
\, c_{1}(\xi_{1}) < c_{2}(\xi_{2}); \\
\,  x^{i}_{2}=0,\;  0 \leq i \leq n;\\
\, 1 \!\leq l \leq \!N\!-\!1; \quad y^{j}_{2}\!>\!0,\;  1 \!\leq\! j \!\leq \! l;\;\\
\, y^{j}_{2}\!=\!0,\;  j\!\in\! \!J_2 \triangleq\!\{l\!+\!1,\ldots, N\!-\!1\}\!\cup\! \{0\},\, \text{ if } J_2\!\neq \!\emptyset.
\end{cases}
\end{equation}
\begin{table}[htbp!]
\begin{center}
{\scriptsize
\begin{tabular}{c c c c c c c c c c c c}
 &$\alpha^{1}$&$\cdots$&$\alpha^{n}$&$\alpha^{0}$&$T^{1}$&$\cdots$&$T^{l}$&$T^{l+1}$&$\cdots$&$T^{N-1}$&$T^{0}$\\
\hline
\rule{0pt}{2.5ex}{\em arc 1}&$\alpha^{1}$&$\cdots$&$\alpha^{n}$&$\alpha^{0}$&$y^{1}_{1}$&$\cdots$&$y^{l}_{1}$&$T^{l+1}$&$\cdots$&$T^{N-1}$&$T^{0}$\\
{\em arc 2}&$0 $&$\cdots$&$ 0$&$0$&$y^{1}_{2}$&$\cdots$&$y^{l}_{2}$&$0$&$ \cdots$&$ 0$&$0$
\end{tabular}
}\caption{Mode 3.}
\end{center}
\end{table}

The CE $z_\alpha$ is of {\em mode 4} if
\begin{equation}\label{mode_5}
c_{1}(\xi_{1})=c_{2}(\xi_{2}).
\end{equation}
\begin{table}[hbtp!]
\begin{center}
{\scriptsize
\begin{tabular}{c c c c c c c c c}
 &$\alpha^{1}$&$\cdots$&$\alpha^{n}$&$\alpha^{0}$&$T^{1}$&$\cdots$&$T^{N-1}$&$T^{0}$\\
\hline
\rule{0pt}{2.5ex}{\em arc 1}&$x^{1}_{1}$&$\cdots$&$x^{n}_{1}$&$x^{0}_{1}$&$y^{1}_{1}$&$\cdots$&$y^{N-1}_{1}$&$y^{0}_{1}$\\
{\em arc 2}&$x^{1}_{2}$&$\cdots$&$x^{n}_{2}$&$x^{0}_{2}$&$y^{1}_{2}$&$\cdots$&$y^{N-1}_{2}$&$y^{0}_{2}$
\end{tabular}
}\caption{Mode 4.}
\end{center}
\end{table}

\begin{lemma}
For all $\alpha \in \mathcal{S}^N$, the CE $z_\alpha$ takes one of the four modes above.
\end{lemma}
\begin{proof}
According to Lemma~\ref{prop:c1_less_c2}, at $z_\alpha$, $c_1(\xi_1)\leq c_2(\xi_2)$, hence there are nonatomic players taking arc 2 only if $c_1(\xi_1)=c_2(\xi_2)$; besides, all the atomic players must put some weight on arc 1;  also, if an atomic player put some weight on arc 2, then all the atomic players larger than her must do so as well. Hence, $z_\alpha$ must take one of the four modes above. 
\end{proof}

\begin{lemma}\label{lm:equiv_decentralization}
If two decentralization strategies $\alpha$ and $\tilde{\alpha}$ in $\cS^{N}$ are such that the total weight on arc 1 is the same at $z_\alpha$ and $z_{\tilde{\alpha}}$, i.e. $\xi_1(z_\alpha)=\xi_1(z_{\tilde{\alpha}})$, then they are equivalent with respect to $T^{-N}$.
\end{lemma}
\begin{proof}
Denote $z_{\tilde{\alpha}}$ by $\tilde{z}$, $x_r(\tilde{z})$ by $\tilde{x}_r$, $y^j_r(\tilde{z})$ by $\tilde{y}^j_r$, and $\xi_r(\tilde{z})$ by $\tilde{\xi}_r$.

If $\xi_1 = \tilde{\xi}_1=\hat{\xi}$, then $U^N(\alpha,T^{-N})=U^N(\tilde{\alpha},T^{-N})=c_1(\hat{\xi})$.

If $\xi_1 \neq \hat{\xi}$, then $y^0_1=T^0$, and according to eq.\eqref{cond_atom}, for each $1 \leq j \leq N-1$, both $y^j_1$ and $\tilde{y}^j_1$ solve the following equation in $w$: \emph{Either $w< T^j$ and $c_{1}(\xi_1)+ w\,c'_1(\xi_{1}) = c_{2}(\xi_2)+ (T^j-w) \,c'_2(\xi_{2})$, or $w=T^j$ and $c_{1}(\xi_1)+ w\,c'_1(\xi_{1}) \leq c_{2}(\xi_2)+ (T^j-w) \,c'_2(\xi_{2})$.}  There is only one solution to this equation, hence $y^j_1=\tilde{y}^j_1$. Similarly, $y^0_1=\tilde{y}^0_1$. Therefore $x_{1}= \tilde{x}_{1}$ and $U^N(\alpha,T^{-N})=U^N(\tilde{\alpha},T^{-N})$.
\end{proof}

\begin{lemma}\label{thm:mode_1_cond}
1. If $z_\alpha$ is of mode 1, then 
\begin{equation}\label{flow_1_a}
x_{1}=T^{N},\quad y_{1}=M^{-N}, \quad \xi_{1}=M.
\end{equation}

2. There exists a strategy $\alpha \in \cS^N$ such that  $z_\alpha$ is of mode 1 if and only if
\begin{equation}\label{cond_1_d}
 H>0; \quad T^{1}\leq H\,\text{ if }\,N >1.
\end{equation}

3. Assume eq.\eqref{cond_1_d}. Define a subset of $\cS^{N}$ by:
\begin{equation*}
\cS^{N}_{1}= \{ \alpha \in \cS^{N}\,\big|\,\alpha = \underline{\alpha}, \text{ or }\, \alpha^{1}\leq H \}.
\end{equation*}
Then, $\alpha\in \cS^{N}_{1}$ if and only if $z_\alpha$ is of mode~1. Furthermore, all the strategies in $\cS^{N}_{1}$ are equivalent to each other with respect to $T^{-N}$.

In particular, the nonatomic strategy $\underline{\alpha}$ is in $\cS^{N}_{1}$.
\end{lemma}
\begin{proof}
Since $\xi_{1}=M$ and $c_{1}(\xi_{1})< c_{2}(\xi_{2})$, one has $c_{1}(M) < c_{2}(0)$ or, equivalently, $H > 0$. If $N \geq 2$, then it follows from eq.\eqref{cond_atom} that, for $ 1 \leq j \leq N-1$, $c_{1}(M) + T^{j}\,c'_{1}(M) \leq c_{2}(0)$ or, equivalently, $T^{j} \leq H$. Similarly, if $n \geq 1$, then for $1 \leq i \leq n$, $\alpha^{i} \leq H$.
\end{proof}

\begin{lemma}\label{thm:mode_2_cond}
1. If $z_\alpha$ is of {\em mode~2 specified by} $k$ and $l$, then 
\begin{align}
& x^{i}_{1} = \frac{\alpha^{i}\,a(\xi_{1})+h(\xi_{1})}{1+a(\xi_{1})},\;  1\leq i \leq k; \quad x_{1} = T^{N}- \frac{\alpha^{[k]}-k\,h(\xi_{1})}{1+a(\xi_{1})}; \label{flow_2_d}\\
& y^{j}_{1} = \frac{T^{j}\,a(\xi_{1}) + h(\xi_{1})}{1+a(\xi_{1})}, \;  1 \leq j \leq l; \quad  y_{1} = M^{-N}- \frac{T^{[l]}-l\,h(\xi_{1})}{1+a(\xi_{1})}; \label{flow_2_c}\\
& \xi_{1}=M -\frac{\alpha^{[k]}+T^{[l]}-(k+l)\,h(\xi_{1})}{1+a(\xi_{1})} = F_{k+l}^{-1}(\alpha^{[k]}+T^{[l]}). \label{flow_2_b}
\end{align}
In particular, if $l\geq 1$, then for all $1\leq j \leq l$, $T^j>h(\xi_1)$.

2. For given $k\in \nit^*$ and $l\in \{0,\ldots,N\!-\!1\}$, there exists a strategy $\alpha \in \cS^N$ such that  $z_\alpha$ is of mode~2 specified by $k$ and $l$ if and only if
\begin{equation}\label{cond_2_c}
\begin{cases}
 F^{-1}_{l}(T^{[l]}) > h^{-1}(T^{l}) \,\text{ if } l\geq 1;\quad T^{N}\! >\! k h(F^{-1}_{l}(T^{[l]})); \\
 T^{N}\! \geq \! F_{k+l}(h^{-1}(T^{l+1})) \!-\! T^{[l]}  \text{ if } N \!\geq\! l+2; \, T^{N} \!>\! F_{0}(\hat{\xi})\! -\! T^{[l]} \text{ if } H\!\leq\! 0.
\end{cases}
\end{equation}

3.  For given $k\in \nit^*$ and $l\in \{0,\ldots,N\!-\!1\}$, assume eq.\eqref{cond_2_c}. Given a real constant $w$ such that
\begin{equation}\label{cond_2_f}
\begin{cases}
 0 \! < \! w \! \leq\!  T^{N};\; w\! >\! k\,h(F^{-1}_{l}(T^{[l]}));\; w \! <\!  F_{k+l}(h^{-1}(T^{l})) \! - \! T^{[l]} \text{ if } l\geq 1;\\
 w \! \geq\!  F_{k+l}(h^{-1}(T^{l+1})) \! -\!  T^{[l]}  \text{ if } l \! \leq\!  N-1;\; w \! >\!  F_{0}(\hat{\xi}) \! -\!  T^{[l]} \text{ if } H\! \leq\!  0,
\end{cases}
\end{equation}
denote $\eta_{1}=F_{k+l}^{-1}(w+T^{[l]})$, and define a subset of $\cS^N$ by:
\begin{align*}
\cS^{N}_{2}(w,T^{-N};k,l)\!=\!\{ \alpha\!\in\! \cS^{N}\,|\,& n\!\geq\! k;\,\alpha^{[k]} \!=\! w; \,\\
& \forall 1\!\leq \!i \!\leq\! k, \alpha^{i}\! >\! h(\eta_{1}); \,\forall k+1 \!\leq\! i \!\leq\! n,\,\alpha^{i}  \!\leq\! h(\eta_{1}) \}.
\end{align*}
Then, $\alpha \in \cS^{N}_2(w,T^{-N};k,l)$ if and only if it induces $z_{\alpha}$ of mode~2 specified by $k$ and $l$ and, at $z_{\alpha}$, the total weight on arc 1, $\xi_1(z_\alpha),$ is $\eta_{1}$. Furthermore, the strategies in $\cS^{N}_{2}(w,T^{-N};k,l)$ are equivalent to each other.

4. SA strategy $w-(k-1)h(\eta_{1})$ is equivalent to the strategies in \linebreak[4]$\cS^{N}_{2}(w,T^{-N};k,l)$.
\end{lemma}
\begin{proof}
Let us prove for the case that $l \geq 1$. The case that $l=0$ is similar.

1. For $1 \leq i \leq k$, eq.\eqref{cond_atom} implies that $c_{1}(\xi_{1})+ x^{i}_{1} \, c'_{1}(\xi_{1}) = c_{2}(\xi_{2}) + x^{i}_{2}\,c'_{2}(\xi_{2})$ or, equivalently, $x^{i}_{1} = \frac{\alpha^{i}\,a(\xi_{1})+h(\xi_{1})}{1+a(\xi_{1})}$. Similarly,
for $1 \leq j \leq l$, $y^{j}_{1} = \frac{T^{j}\,a(\xi_{1}) + h(\xi_{1})}{1+a(\xi_{1})}$.

The rest of the results in eq.\eqref{flow_2_d}, eq.\eqref{flow_2_c} and eq.\eqref{flow_2_b} can then be easily obtained. In particular,
\begin{equation}
F_{k+l}(\xi_{1}) = \alpha^{[k]}+T^{[l]}. \label{flow_2_a}
\end{equation}
Since $F_{k+l}$ is strictly decreasing, $\xi_{1}=F_{k+l}^{-1}(\alpha^{[k]}+T^{[l]})$.

2--3. For $1 \leq i \leq k$, the fact that $x^{i}_{1} = \frac{\alpha^{i}\,a(\xi_{1})+h(\xi_{1})}{1+a(\xi_{1})}$ and $0 < x^{i}_{1} < \alpha^{i}$ yields (i) $\alpha^{i}> -h(\xi_{1})/a(\xi_{1})$, which is always true, because $h(\xi_{1})\geq 0$ and $a(\xi_{1})>0$; and (ii)
\begin{equation} \label{cond_2_alpha}
\alpha^{i} > h(\xi_{1}).
\end{equation}
Similarly as for eq.\eqref{cond_2_alpha}, one has $T^{j} > h(\xi_{1})$ for all $1 \leq j \leq l$.

Four constraints on $\xi_{1}$ can be deduced.

(i) Eq.\eqref{cond_2_alpha} implies
\begin{equation}\label{cond_2_g}
 \alpha^{[k]}> k \,h(\xi_{1})
\end{equation}
or, equivalently, according to eq.\eqref{flow_2_a}, $F_{k+l}(\xi_{1})=(M-\xi_{1})(1+a(\xi_{1}))+(k+l)\,h(\xi_{1})>k\,h(\xi_{1}) + T^{[l]} \Rightarrow F_{l}(\xi_{1})=(M-\xi_{1})\bigl(1+a(\xi_{1})\bigr)+l\,h(\xi_{1})>T^{[l]}$. Therefore, $\xi_{1} < F^{-1}_{l}(T^{[l]})$. 

(ii)  The fact that $T^{j} > h(\xi_{1})$ for all $1 \leq j \leq l$ implies $\xi_{1} > h^{-1}(T^{l})$.

(iii) If $l\!<\!N\!-\!1$, then for $l\!+\!1 \!\leq\! j \!\leq\!  N\!-\!1$, according to eq.\eqref{cond_atom}, $c_{1}(\xi_{1}) + T^{j}\,c'_{1}(\xi_{1}) \leq c_{2}(\xi_{2})$ or equivalently $T^{j} \leq h(\xi_{1})$, which further implies $\xi_{1} \leq  h^{-1}(T^{l+1})$.

(iv) If $H\leq 0$, then according to Lemma~\ref{prop:c1_less_c2}, $\xi_{1}< \hat{\xi}$.

These four constraints on $\xi_{1}$ (i.e. $\xi_{1} \!<\! F^{-1}_{l}(T^{[l]})$,  $\xi_{1} \!>\! h^{-1}(T^{l})$,  $\xi_{1} \!\leq\!  h^{-1}(T^{l+1})$ and $\xi_{1}< \hat{\xi}$) together with eq.\eqref{cond_2_alpha}, eq.\eqref{flow_2_a} and eq.\eqref{cond_2_g} imply that
\begin{align}
& T^{N} \geq  \alpha^{[k]} > k h(\xi_{1}) > k\,h(F^{-1}_{l}(T^{[l]})); \label{cond_2_a}\\
& \alpha^{[k]} = F_{k+l}(\xi_{1}) - T^{[l]} < F_{k+l}(h^{-1}(T^{l})) - T^{[l]}; \label{cond_2_b}\\
& T^{N} \geq  \alpha^{[k]} = F_{k+l}(\xi_{1}) - T^{[l]} \geq F_{k+l}(h^{-1}(T^{l+1})) - T^{[l]},\; \text{ if }\, l < N-1; \label{cond_2_d} \\
& T^{N} \geq  \alpha^{[k]} = F_{k+l}(\xi_{1}) - T^{[l]} \geq F_{0}(\hat{\xi}) - T^{[l]}, \; \text{ if }\, H\leq 0. \label{cond_2_k}
\end{align}
Equations \eqref{cond_2_a}-\eqref{cond_2_k} yield eq.\eqref{cond_2_f}.

Besides, eq.\eqref{cond_2_a} and eq.\eqref{cond_2_b} imply that
\begin{align*}
& k\,h(F^{-1}_{l}(T^{[l]})) < F_{k+l}(h^{-1}(T^{l})) - T^{[l]} \\
&\Rightarrow T^{[l]} + k\,h(F^{-1}_{l}(T^{[l]})) < F_{l}(h^{-1}(T^{l})) + k\,T^{l} \\
&\Rightarrow  F^{-1}_{l}(T^{[l]}) > h^{-1}(T^{l}).
\end{align*}
This result, together with eq.\eqref{cond_2_a}-eq.\eqref{cond_2_k}, proves eq.\eqref{cond_2_c}.

If $n > k$, then for $k+1 \leq i \leq n$, according to eq.\eqref{cond_atom}, $c_{1}(\xi_{1}) + \alpha^{i}\,c'_{1}(\xi_{1}) \leq c_{2}(\xi_{2})$ or, equivalently, $\alpha^{i} \leq h(\xi_{1})$.

4. In five steps, let us show that, if the conditions in eq.\eqref{cond_2_c} and those in eq.\eqref{cond_2_f} are satisfied for $k$ and $w$, then they are also satisfied when $k$ is replaced by 1, and $w$ is replaced by $w-(k-1)h(\eta_{1})$.

1) If $h(F^{-1}_{l}(T^{[l]}))<0$, then $ T^{p} \geq w > h(F^{-1}_{l}(T^{[l]}))$. If $h(F^{-1}_{l}(T^{[l]}))>0$, $ T^{N} \geq w > k\,h(F^{-1}_{l}(T^{[l]})) > h(F^{-1}_{l}(T^{[l]}))$.

2) If $l<N-1$, then $T^{N} \geq  F_{k+l}(h^{-1}(T^{l+1})) - T^{[l]} = F_{1+l}(h^{-1}(T^{l+1})) - T^{[l]} + (k-1)T^{l+1} \geq F_{1+l}(h^{-1}(T^{l+1})) - T^{[l]}$, where the second inequality is due to the fact that $k\geq 1$.

3) The fact that $\eta_{1}>h^{-1}(T^{l})$ implies that $F_{1+l}(\eta_{1})< F_{1+l}(h^{-1}(T^{l}))$. Besides, by the definition of $\eta_{1}$,
$ w - (k-1) h(\eta_{1}) = F_{k+l}(\eta_{1}) - T^{[l]} - (k-1)h(\eta_{1}) = F_{1+l}(\eta_{1}) - T^{[l]}$. Therefore,  $w-(k-1)h(\eta_{1})< F_{1+l}(h^{-1}(T^{l})) - T^{[l]}$.

4) If $l\!<\!N-1$, the relation $\eta_{1}\!\leq\! h^{-1}(T^{l+1})$ implies that $F_{1+l}(\eta_{1})\! \geq \!F_{1+l}(h^{-1}(T^{l+1}))$. As a result, $w-(k-1)h(\eta_{1})=F_{1+l}(\eta_{1}) - T^{[l]} \geq F_{1+l}(h^{-1}(T^{l+1})) - T^{[l]}$.

5) If $H\leq 0$, the relation $\eta_{1}\leq \hat{\eta}$ implies that $w-(k-1)h(\eta_{1}) =F_{1+l}(\eta_{1}) - T^{[l]}\geq F_{0}(\hat{\eta})- T^{[l]}$.

Therefore, the CE $z_{w-(k-1)h(\eta_{1})}$ induced by SA strategy $w-(k-1)h(\eta_{1})$ is of mode~3 specified by $1$ and $l$. The definition of $\eta_{1}$ implies that $w+T^{[l]}=F_{k+l}(\eta_{1})=F_{1+l}(\eta_{1})+(k-1)h(\eta_{1})$. Hence, $w - (k-1)h(\eta_{1}) +T^{[l]}=F_{1+l}(\eta_{1})$ or, equivalently, $\eta_{1}=F^{-1}_{1+l}(w - (k-1)h(\eta_{1}))$. According to eq.\eqref{flow_2_b}, the total weight on arc~1 at $z_{w-(k-1)h(\eta_{1})}$ is also $\eta_{1}$, which means that SA strategy $w-(k-1)h(\eta_{1})$ is equivalent to the strategies in $\cS^{p}_{2}(w,T^{-N};k,l)$, by applying Lemma~\ref{lm:equiv_decentralization}.
\end{proof}

\begin{lemma}\label{thm:mode_3_cond}
1. If $z_{\alpha}$ is of mode~3 and specified by $l$, then
\begin{align}
& x_{1} \! = \! T^{N}; \;y^{j}_{1}\!  = \! \frac{T^{j}\,a(\xi_{1}) \! +\!  h(\xi_{1})}{1\! +\! a(\xi_{1})},\,   1\! \leq\!  j \! \leq\!  l; \; y_{1} \! =\!  M^{-N}\! -\!  \frac{T^{[l]}\! -\! l\,h(\xi_{1})}{1\! +\! a(\xi_{1})}; \label{flow3_2}\\
& \xi_{1}\!  =\!  M \! -\!  \frac{T^{[l]}\! -\! l\,h(\xi_{1})}{1\! +\! a(\xi_{1})}\! =\! F^{-1}_{l}(T^{[l]}). \label{flow3_3}
\end{align}
In particular, for all $1\leq j \leq l$, $T^j>h(\xi_1)$.
 
2. If $N>1$, for given $l\in \nit^*$ and $l<N$, there exists a strategy $\alpha\in \cS^{N}$ such that $z_{\alpha}$ is of mode~3 and specified by $l$ if and only if
\begin{equation}\label{cond_3_b}
\begin{cases}
 F^{-1}_{l}(T^{[l]}) > h^{-1}(T^{l}); \;\\
 F^{-1}_{l}(T^{[l]}) \leq h^{-1}(T^{l+1})\, \text{if } l\leq N-2; \; \\
T^{[l]} > F_{0}(\hat{\xi})\, \text{ if } H\leq 0.
\end{cases}
\end{equation}

3. If $N>1$, for given $l\in \nit^*$ and $l<N$, assume eq.\eqref{cond_3_b}. Define a subset of $\cS^{N}$ by:
\begin{align*}
\cS^{N}_{3}(T^{-N};l)=\{& \,\alpha\in \cS^{N}\,|\, \alpha = \underline{\alpha}, \text{ or } \alpha^{1} \leq h(F^{-1}_{l}(T^{[l]})) \,\}.
\end{align*}
Then, $\alpha \in \cS^{N}_{3}(T^{-N};l)$ if and only if it induces $z_{\alpha}$ of mode~3 specified by $l$ and, at $z_{\alpha}$, the total weight on arc 1 is $F^{-1}_{l}(T^{[l]})$. Furthermore, the strategies in $\cS^{N}_{3}(T^{-N};l)$ are equivalent to each other.

In particular, the nonatomic strategy $\underline{\alpha}$ is in $\cS^{N}_{3}(T^{-N};l)$.
\end{lemma}
\begin{proof}
Similar to the proof of Lemma~\ref{thm:mode_2_cond}.
\end{proof}

\begin{lemma}\label{thm:mode_5_cond}
1. If $z_\alpha$ is of mode 4, then
\begin{equation}
x^{i}_{1}=\frac{A\alpha^{i}}{1+A},\; 1 \!\leq \!i \!\leq\! n, \,\text{if } n \! \geq \!1; \;\, y^{j}_{1}=\frac{AT^{j}}{1+A},\; 1\!\leq\! j\! \leq\! N\!-\!1, \,\text{if } N \!>\! 1; \;\,  \xi_1 = \hat{\xi}. \label{flow5b}
\end{equation}

2. There exists a strategy $\alpha \in \cS^N$ such that $z_{\alpha}$ is of mode 4 if and only if
\begin{equation}\label{cond_5_a}
 H\leq 0; \quad  T^{[N-1]} \leq F_{0}(\hat{\xi}), \text{ if } N > 1.
\end{equation}

3. Assume eq.\eqref{cond_5_a}. Define a subset of $\cS^{N}$ by:
\begin{equation*}
\cS^{N}_{4}=\{ \,\alpha=(\alpha^i)_{i=0}^{n}\in \cS^{N}\,|\, \alpha = \underline{\alpha}, \text{ or } \alpha^{[n]} \leq  F_{0}(\hat{\xi}) - T^{[N-1]} \}.
\end{equation*}
Then $\alpha\in \cS^N_4$ if and only if $z_{\alpha}$ is of mode 4. Furthermore, the strategies in $\cS^{N}_{4}$ are equivalent to each other.

In particular, the nonatomic strategy $\underline{\alpha}$ is in $\cS^{N}_{4}$.
\end{lemma}
\begin{proof}
Since the two arcs have equal cost, $\xi_{1}=\hat{\xi}$ and, according to eq.\eqref{prop:xi}, $H \leq 0$. Flows eq.\eqref{flow5b} can be deduced from eq.\eqref{cond_atom}. And they imply that
\begin{equation} \label{flow_5_d}
x^{0}_{1}+y^{0}_{1}=\xi_{1}-\sum^{n}_{i}x^{i}_{1}-\sum^{N-1}_{j}y^{j}_{1}=\hat{\xi}-\frac{A(\alpha^{[n]}+T^{[N-1]})}{1+A}.
\end{equation}

It follows from $0 \leq x^{0}_{1}\leq \alpha^{0}$ and $0 \leq y^{0}_{1}\leq T^{0}$ that $0 \leq x^{0}_{1}+y^{0}_{1} \leq T^{N} - \alpha^{[n]} + M^{-N} -T^{[N-1]}=M - \alpha^{[n]} - T^{[N-1]}$. One deduces, by considering eq.\eqref{flow_5_d}, that $\alpha^{[n]} + T^{[N-1]} \leq \hat{\xi}\cdot \frac{1+A}{A}$ and $\alpha^{[n]} + T^{[N-1]} \leq (M-\hat{\xi})(1+A)$. But, $\hat{\xi}\cdot \frac{1+A}{A}\geq (M-\hat{\xi})(1+A)$ because $(M-\hat{\xi})(1+A) \leq M \leq \hat{\xi} \frac{1+A}{A}$, which is itself due to eq.\eqref{eq:assp_c1}. Therefore, $\alpha^{[n]} + T^{[N-1]} \leq (M-\hat{\xi})(1+A)=F_{0}(\hat{\xi})$.
\end{proof}
\subsection{Lemmas and proofs}
\begin{lemma}\label{lm:reply_0}
Suppose that $\alpha,\tilde{\alpha}\in \cS^N$ are two decentralization strategies of atomic player~$N$, and $z_{\alpha}, z_{\tilde{\alpha}}$ respectively the CE of game $\Gamma(\alpha, T^{-N})$ and $\Gamma(\tilde{\alpha}, T^{-N})$. If $z_{\alpha}$ is of mode~2, and $x_1(z_{\tilde{\alpha}})>x_{1}(z_{\alpha})$, $y_1(z_{\tilde{\alpha}})\geq y_{1}(z_{\alpha})$, then
\begin{equation*}
U^{N}(\tilde{\alpha},T^{-N}) > U^{N}(\alpha,T^{-N}).
\end{equation*}
\end{lemma}
\begin{proof}
Denote $z_\alpha$ by $z$, $z_{\tilde{\alpha}}$ by $\tilde{z}$, $x_r(z)$ by $x_r$, $x_r(\tilde{z})$ by $\tilde{x}_r$, $y_r(z)$ by $y_r$, $y_r(\tilde{z})$ by $\tilde{y}_r$, $\xi_r(z)$ by $\xi_r$ and $\xi_r(\tilde{z})$ by $\tilde{\xi}_r$.

Suppose that $x$ is specified by $k\in \mathbb{N}^{*}$ (and $l\in \mathbb{N}^{*}$ in the case that $x$ is of mode 2). According to eq.\eqref{cond_atom}, 
\begin{equation}\label{eq:88}
c_{1}(\xi_{1}) + x^{i}_{1}\,c'_{1}(\xi_{1}) = c_{2}(\xi_{2}) + x^{i}_{2}\,c'_{2}(\xi_{2}), \quad 1 \leq i \leq k.
\end{equation}
Summing eq.\eqref{eq:88} leads to $ k c_{1}(\xi_{1}) + [ x_{1} - (T^{N} - \alpha^{[k]})] c'_{1}(\xi_{1})= k c_{2}(\xi_{2}) + x_{2} c'_{2}(\xi_{2})$. Consequently, $c_{1}(\xi_{1}) + x_{1} c'_{1}(\xi_{1}) = c_{2}(\xi_{2}) + x_{2} c'_{2}(\xi_{2}) + (k-1) [c_{2}(\xi_{2})-c_{1}(\xi_{1})] + (T^{N}-\alpha^{[k]}) c'_{1}(\xi_{1})$. Since $c_{1}(\xi_{1})\leq c_{2}(\xi_{2})$ by Lemma~\ref{prop:c1_less_c2}, one deduces that $c_{1}(\xi_{1}) + x_{1} c'_{1}(\xi_{1}) \geq  c_{2}(\xi_{2}) + x_{2} c'_{2}(\xi_{2})$. Moreover, there exists $B>0$ such that
\begin{equation}\label{eq:reply_0_1}
c_{1}(\xi_{1}) + x_{1}\,c'_{1}(\xi_{1}) \geq B \geq \,  c_{2}(\xi_{2}) + x_{2}\,c'_{2}(\xi_{2}).
\end{equation}
Assumption~\ref{assp:base_chap3}, eq.\eqref{eq:reply_0_1} and the fact that $c_{1}(\tilde{\xi}_{1})\leq c_{2}(\tilde{\xi}_{2})$ (still by Lemma~\ref{prop:c1_less_c2}) imply that, for all $s\in (x_{1}, \tilde{x}_1]$ and $t\in [\tilde{x}_2, x_{2})$,
\begin{align}
& c_{1}(s+y_{1}) +s\,c'_{1}(s+y_{1})  > B >   c_{2}(t+y_{2}) + t\,c'_{2}(t+y_{2}),
\label{eq:reply_0_3}\\
& c_{1}(s+y_{1})\leq c_{1}(\tilde{\xi}_{1})\leq c_{2}(\tilde{\xi}_{2}) \leq c_{2}(t+y_{2}).\nonumber
\end{align}
They further imply that $s\,c'_{1}(s+y_{1})>t\,c'_{2}(t+y_{2})$ and, moreover, there exists $C>0$ such that for any $p\in (y_{1}, M^{-N}]$ and $q\in [0, y_{2})$,
\begin{equation}\label{eq:reply_0_4}
s\,c'_{1}(s+p)  > C >  t\,c'_{2}(t+q).
\end{equation}

Let us compare $U^{N}(\tilde{\alpha},T^{-N})$ and $U^{N}(\alpha,T^{-N})$:
\begin{align*}
& U^{N}(\tilde{\alpha},T^{-N}) - U^{N}(\alpha,T^{-N})  =  [ \tilde{x}_{1}\,c_{1}(\tilde{\xi}_{1}) + \tilde{x}_{2}\, c_{2}(\tilde{\xi}_{2}) ] - [x_{1}\,c_{1}(\xi_{1}) + x_{2}\, c_{2}(\xi_{2})] \\
& = [ \tilde{x}_{1}\,c_{1}(\tilde{\xi}_{1}) - \tilde{x}_{1}\,c_{1}(\tilde{x}_{1}+y_{1})] + [ \tilde{x}_{1}\,c_{1}(\tilde{x}_{1}+y_{1}) - x_{1}\,c_{1}(\xi_{1}) ] - [  \tilde{x}_{2}\,c_{2}(\tilde{x}_{2}+y_{2}) \\
& ~~-\tilde{x}_{2}\, c_{2}(\tilde{\xi}_{2}) ] - [ x_{2}\, c_{2}(\xi_{2}) -\tilde{x}_{2}\, c_{2}(\tilde{x}_{2}+y_{2})]\\
%
%
& = \int^{\tilde{y}_{1}}_{y_{1}}\tilde{x}_{1} c'_{1}(\tilde{x}_{1}+s) \text{d}s + \int^{\tilde{x}_{1}}_{x_{1}} [ c_{1}(s+y_{1}) +s c'_{1}(s+y_{1}) ] \text{d}s \\
&~~  -\int^{y_{2}}_{\tilde{y}_{2}}\tilde{x}_{2} c'_{2}(\tilde{x}_{2}+t) \text{d}t- \int^{x_{2}}_{\tilde{x}_{2}} [ c_{2}(t+y_{2}) +t c'_{2}(t+y_{2}) ] \text{d}t\\
& > (\tilde{y}_{1}-y_{1}) C + (\tilde{x}_{1}-x_{1}) B - (y_{2}-\tilde{y}_{2}) C - (x_{2}-\tilde{x}_{2}) B  = 0,
\end{align*}
where the inequality is due to eq.\eqref{eq:reply_0_3} and eq.\eqref{eq:reply_0_4}.
\end{proof}

\begin{lemma}\label{lm:trivial_strategy}
For $l=0$ (if $N>1$) and for $1 \leq l \leq N-2$ (if $N>2$), define
\begin{equation*}
B_{l}=F_{l+1}(h^{-1}(T^{l+1})) - T^{[l]}.
\end{equation*}

Suppose that either (i) $H>0$, $N >1$ and $T^{1}> H$, or (ii) $H\leq 0$, $N \geq 2$, $T^{[N-1]}\geq  F_{0}(\hat{\xi})$. Then,

1. The CE $z_{\underline{\alpha}}$ induced by atomic player~$N$'s nonatomic strategy $\underline{\alpha}$ is of mode~3 specified by $l_{0}$, $l_{0}$ being the unique number in $\{1,\ldots, N\!-\!1\}$ that meets the conditions in eq.\eqref{cond_3_b}.

2. \begin{align}
& F^{-1}_{1}(T^{[1]}) >  F^{-1}_{2}(T^{[2]}) > \cdots > F^{-1}_{l_{0}}(T^{[l_{0}]}); \label{rlt_1}\\
& F^{-1}_{l}(T^{[l]})> h^{-1}(T^{l}),\quad 1 \leq l \leq l_{0}. \label{rlt_3}
\end{align}

If $l_{0}<N-1$, then
\begin{align}
& F^{-1}_{l_{0}}(T^{[l_{0}]}) \leq F^{-1}_{l_{0}+1}(T^{[l_{0}+1]}) \leq \cdots \leq  F^{-1}_{N-1}(T^{[N-1]}), \, \text{ if } l_0>1; \label{rlt_2}\\
& F^{-1}_{l}(T^{[l]}) \leq h^{-1}(T^{l}), \quad l_{0}+1 \leq l \leq N-1. \label{rlt_4}
\end{align}

3. For all $k\in \mathbb{N}$,
\begin{equation}\label{rlt_5}
F_{k+l}(h^{-1}(T^{l+1})) - T^{[l]}
\begin{cases}
>\! k\,h(F^{-1}_{l}(T^{[l]})),  \,\text{ if } l_{0}\!>\!1\,\text{ and }\, 1 \!\leq\! l\! \leq\! l_{0}\!-\!1,\\
\leq \!k\,h(F^{-1}_{l}(T^{[l]})), \, \text{ if } l_{0}\! <\! N\!-\!1\,\text{ and }\,l \!=\! l_{0}.
\end{cases}
\end{equation}
In particular, if $l_{0}>1$, then $B_{l}>h(F^{-1}_{l}(T^{[l]}))$ for all $1\leq l \leq l_{0}-1$.

4. If $l_{0}>1$, then for all $k\in \mathbb{N}$ and $1 \leq l \leq  l_{0}-1$, $F_{k+l}(h^{-1}(T^{l+1})) - T^{[l]} > k h(F^{-1}_{l_{0}}(T^{[l_{0}]}))$.

5. If $l_{0}\!>1$, then $B_{0} \!\geq\! B_{1} \!\geq\!\cdots\! \geq\! B_{l_{0}-1}$. For $l\!\in\! \{1,\ldots, l_{0}\! -\!1 \}$, $B_{l}\! = \!B_{l-1}$ if and only if $T^{l+1}\!=\!T^{l}$.

6. $B_{l_{0}-1}>h(F^{-1}_{l_{0}}(T^{[l_{0}]}))$.

7. For all strategy $\alpha\in \cS^{N}$, one and only one of the following is true:

(i) $z_\alpha$ is of mode~3 and specified by $l_{0}$; besides, $\alpha$ is equivalent to $\underline{\alpha}$.

(ii) $z_{\alpha}$ is of mode~2 and specified by some $k\in \mathbb{N}^{*}$ and some $l\in \{0, 1,\ldots, l_{0} \}$.
\end{lemma}
\begin{proof}
Denote $z_{\underline{\alpha}}$ simply by $\underline{z}$.

1. The CE $\underline{z}$ is not of mode~1 when $H>0$ because of eq.\eqref{cond_1_d}, and is not of mode 4 when $H\leq 0$ because of eq.\eqref{cond_5_a}. It is not of mode~2 because $n=0$. Thus, it is of mode~3. Lemma~\ref{thm:mode_3_cond} shows that $\underline{z}$ must be specified by some unique $l_{0}\in \{1,\ldots, N-1 \}$  satisfying eq.\eqref{cond_3_b}.

2. One has only to prove for the case where $N>2$. First, suppose that $l_{0}<N-1$. Let us prove eq.\eqref{rlt_2} and eq.\eqref{rlt_4} by induction.

The fact that $l_{0}\neq N-1$ implies that $N-1$ does not meet all the four conditions in eq.\eqref{cond_3_b}. However, in this case, only the second condition can be violated. Thus, $F^{-1}_{N-1}(T^{[N-1]}) \leq h^{-1}(T^{N-1})$. Now, suppose that for some $l\in \{l_{0}+1,\ldots,  N-1 \}$,
\begin{align*}
& F^{-1}_{l}(T^{[l]}) \,\leq\, F^{-1}_{l+1}(T^{[l+1]}) \leq \cdots \leq  F^{-1}_{N-1}(T^{[N-1]});\;\; \\ 
& F^{-1}_{p}(T^{[p]})\, \leq\, h^{-1}(T^{p}),\, \forall p \in \{ l ,\ldots, N\!-\!1\}.
\end{align*}

The relation $F^{-1}_{l}(T^{[l]}) \leq h^{-1}(T^{l})$ implies that $h(F^{-1}_{l}(T^{[l]})) \geq T^{l}$ and, consequently, $T^{[l]} = F_{l} ( F^{-1}_{l}(T^{[l]})) = F_{l-1} ( F^{-1}_{l}(T^{[l]})) + h(F^{-1}_{l}(T^{[l]})) \geq F_{l-1} ( F^{-1}_{l}(T^{[l]})) + T^{l}$. This implies that $T^{[l-1]} \geq F_{l-1} ( F^{-1}_{l}(T^{[l]}))$ and, as a consequence, $F^{-1}_{l-1}(T^{[l-1]}) \leq F^{-1}_{l}(T^{[l]}) \leq h^{-1}(T^{l})$. In particular, $F^{-1}_{l-1}(T^{[l-1]}) \leq F^{-1}_{l}(T^{[l]}) $ and $F^{-1}_{l-1}(T^{[l-1]}) \leq h^{-1}(T^{l})$.

If $F^{-1}_{l-1}(T^{[l-1]}) > h^{-1}(T^{l-1})$, then all the conditions in eq.\eqref{cond_3_b} are satisfied so that $l_{0}=l-1$. Otherwise, one continues the induction by considering $l-1$. In this way, eq.\eqref{rlt_2} and eq.\eqref{rlt_4} are proved.

Next, let us prove eq.\eqref{rlt_1} and eq.\eqref{rlt_3} by induction.

According to eq.\eqref{cond_3_b}, $F^{-1}_{l_{0}}(T^{[l_{0}]}) > h^{-1}(T^{l_{0}})$. Suppose that $l_{0}>1$ and, for some $l\in  \{ 2,\ldots,  l_{0} \}$,
\begin{align*}
& F^{-1}_{l}(T^{[l]}) >  F^{-1}_{l+1}(T^{[l+1]}) > \cdots > F^{-1}_{l_{0}}(T^{[l_{0}]});\;\;\\
& F^{-1}_{p}(T^{[p]}) > h^{-1}(T^{p}),\; \forall p\in \{ l, \ldots, l_{0} \}.
\end{align*}

If $F^{-1}_{l-1}(T^{[l-1]}) \leq  F^{-1}_{l}(T^{[l]})$, then $ F_{l-1}( F^{-1}_{l-1}(T^{[l-1]})) \geq F_{l-1}( F^{-1}_{l}(T^{[l]}))$, i.e. $T^{[l-1]} \geq F_{l}( F^{-1}_{l}(T^{[l]})) - h( F^{-1}_{l}(T^{[l]})) = T^{[l]}- h( F^{-1}_{l}(T^{[l]}))$ which implies $h( F^{-1}_{l}(T^{[l]})) \geq T^{l}$ or equivalently $F^{-1}_{l}(T^{[l]}) \leq h^{-1}(T^{l})$. It contradicts the hypothesis that $F^{-1}_{l}(T^{[l]}) > h^{-1}(T^{l})$. Therefore, $F^{-1}_{l-1}(T^{[l-1]}) >  F^{-1}_{l}(T^{[l]})$. Furthermore, $F^{-1}_{l-1}(T^{[l-1]}) >  F^{-1}_{l}(T^{[l]})> h^{-1}(T^{l}) \geq h^{-1}(T^{l-1})$. These prove eq.\eqref{rlt_1} and eq.\eqref{rlt_3}.

3. For all $k\in \mathbb{N}$,
\begin{align}
& F_{k+l}(h^{-1}(T^{l+1})) - T^{[l]} - k\,h(F^{-1}_{l}(T^{[l]})) \nonumber \\
& = \, F_{l}(h^{-1}(T^{l+1})) - T^{[l]} + k [ T^{l+1} - h(F^{-1}_{l}(T^{[l]}))] \nonumber \\
& \begin{cases}
\, > 0,\; & \text{if }  h^{-1}(T^{l+1}) < F^{-1}_{l}(T^{[l]}),\\
\, \leq 0, & \text{if }  h^{-1}(T^{l+1}) \geq F^{-1}_{l}(T^{[l]}).
\end{cases} \label{eq:outil}
\end{align}

If $l_{0}>1$, then, for all $l\in \{1,\ldots, l_{0}-1 \}$, eq.\eqref{rlt_1} and eq.\eqref{rlt_3} show that $F^{-1}_{l}(T^{[l]})>F^{-1}_{l+1}(T^{[l+1]})>h^{-1}(T^{l+1})$. 
 If $l_{0} < m$, then, according to eq.\eqref{cond_3_b}, $F^{-1}_{l_{0}}(T^{[l_{0}]}) \leq h^{-1}(T^{l_{0}+1})$. These two inequalities and eq.\eqref{eq:outil} lead to the conclusion.

4. On the one hand, for $l\in \{ 1,\ldots, l_{0}\!-\!1\}$, it is proven in the previous statement that $F^{-1}_{l}(T^{[l]})>h^{-1}(T^{l+1})$ and, consequently, $T^{[l]}<F_{l}(h^{-1}(T^{l+1}))$. On the other hand, eq.\eqref{cond_3_b} implies that $h^{-1}(T^{l_{0}})> F^{-1}_{l_{0}}(T^{[l_{0}]})$ or equivalently $T^{l_{0}}< h(F^{-1}_{l_{0}}(T^{[l_{0}]}))$, thus, $T^{l+1}> h(F^{-1}_{l_{0}}(T^{[l_{0}]}))\geq T^{l_{0}}$.

These two results imply that $  F_{k+l}(h^{-1}(T^{l+1})) - T^{[l]} - k\,h(F^{-1}_{l_{0}}(T^{[l_{0}]})) = \, F_{l}(h^{-1}(T^{l+1})) - T^{[l]} + k\,[T^{l+1}- h(F^{-1}_{l_{0}}(T^{[l_{0}]}))] >0$, which concludes.

5. For $l\in \{ 0, \ldots, l_{0}-2 \}$, $B_{l}-B_{l+1} = [ F_{l+1}(h^{-1}(T^{l+1})) - T^{[l]} ] -\linebreak[4] [ F_{l+2}(h^{-1}(T^{l+2})) - T^{[l+1]} ] = [ F_{l+2}(h^{-1}(T^{l+1})) - T^{l+1} - T^{[l]} ] - [ F_{l+2}(h^{-1}(T^{l+2})) - T^{[l]} - T^{l+1} ] = F_{l+2}(h^{-1}(T^{l+1})) - F_{l+2}(h^{-1}(T^{l+2})) \geq 0 $, because $T^{l+1}\geq T^{l+2}$. Clearly, equality holds if and only if $T^{l+1}= T^{l+2}$.

6. According to eq.\eqref{cond_3_b}, $F^{-1}_{l_{0}}(T^{[l_{0}]}) > h^{-1}(T^{l_{0}})$ or equivalently $T^{[l_{0}]} < F_{l_{0}}(h^{-1}(T^{l_{0}}))$. Hence, $F_{l_{0}}(h^{-1}(T^{l_{0}})) > T^{[l_{0}]} = T^{[l_{0}-1]} + T^{l_{0}} > T^{[l_{0}-1]} + h(F^{-1}_{l_{0}}(T^{[l_{0}]}))$ and, consequently, $B_{l_{0}-1} = F_{l_{0}}(h^{-1}(T^{l_{0}})) - T^{[l_{0}-1]} > h(F^{-1}_{l_{0}}(T^{[l_{0}]}))$.

7. Given an arbitrary strategy $\alpha\in \cS^{p}$. Because $T^{1}>H$ in the case where $H>0$, and $T^{[N-1]}\geq F_{0}(\hat{\xi})$ in the case where $H\leq 0$, $z_{\alpha}$ cannot be of mode~1 or mode 4 according to Lemmas~\ref{thm:mode_1_cond} and \ref{thm:mode_5_cond}. If $z_{\alpha}$ is of mode~3 and specified by $l$, then $l$ meets the conditions in eq.\eqref{cond_3_b}, i.e.  $l=l_{0}$. 
If $z_{\alpha}$ is of mode~2 and specified by $k$ and $l$, then $F^{-1}_{l}(T^{[l]}) > h^{-1}(T^{l})$ by eq.\eqref{cond_2_c}. According to eq.\eqref{rlt_3} and eq.\eqref{rlt_4}, $l\in \{ 1 ,\ldots, l_{0} \}$.
\end{proof}

\begin{lemma}[Nonatomic case]\label{lm:all_best}
Suppose that one of the following holds:

\noindent(i) $H\!>\!0, T^{N}\!\leq\! H$, and $T^{1}\!\leq\! H$ if $N>1$;

\noindent(ii) $H\!>\!0, N\!>\!1, T^{1}\!>\! H, T^{N}\!\leq\! h(F^{-1}_{l_{0}}(T^{[l_{0}]}))$, where $l_{0}$ is the one in Lemma~\ref{lm:trivial_strategy};
 
\noindent(iii) $H\!\leq\! 0$, $T^N+T^{[N-1]}\!\leq\! F_{0}(\hat{\xi})$;

\noindent(iv) $H\!\leq\! 0$, $N\! \geq\! 2$, $T^{[N-1]}\!>\! F_{0}(\hat{\xi})$, $T^{N}\!\leq\! h(F^{-1}_{l_{0}}(T^{[l_{0}]}))$, where $l_{0}$ is the one in Lemma~\ref{lm:trivial_strategy}. Then,

1. All the strategies in $\cS^{N}$ are equivalent to the nonatomic strategy $\underline{\alpha}$. In particular, every strategy of atomic player~$N$ is optimal.

For all SA strategy $s\in [0,T^N]$, $z_s=z_{\bar{\alpha}}$, the CE induced by the trivial strategy $\bar{\alpha}$ or equivalently the CE of the original game $\Gamma(T)$ without decentralization.

2. $z_{\underline{\alpha}}$ is of mode~1 in case (i), of mode~4 in case (iii), and of mode~3 and specified by $l_{0}$ in cases (ii) and (iv).
\end{lemma}
\begin{proof}
The results follow from Lemmas~\ref{thm:mode_1_cond}, \ref{thm:mode_3_cond}, \ref{thm:mode_5_cond} and Lemma~\ref{lm:trivial_strategy} (2).
\end{proof}
\begin{lemma}[Trivial case]\label{lm:lin_1}
Suppose that either (i) $H\!>\!0, T^{N}\!> \!H$,  and $T^{1}\!\leq\! H$ if $N\!>\!1$, or (ii) $H \!\leq\! 0$, $N\!=\!1$ and $T^{N}\!>\! F_{0}(\hat{\xi})$. Then

1. $z_{\bar{\alpha}}$ is of mode 2 specified by 1 and 0.  

2. Atomic player~$N$'s unique optimal decentralization strategy is the trivial one $\bar{\alpha}$, i.e. not decentralizing.

3. For SA strategies $s\in [0,T^N]$, $\xi_1(z_s)$ and $x_1(z_s)$ are continuous and non-increasing in $s$, and $y_1(z_s)$ is continuous and non-decreasing in $s$.
\end{lemma}
\begin{proof}
Denote $z_{\bar{\alpha}}$ by $\bar{z}$, $x_r(\bar{z})$ by $\bar{x}_r$, $y_r(\bar{z})$ by $\bar{y}_r$, and $\xi_r(\bar{z})$ by $\bar{\xi}_r$. Consider the case $H>0$ only. The proof for the case $H\leq 0$ is similar. 

1. For any strategy $\alpha\in \cS^{N}$, $z_{\alpha}$ cannot be of mode 4, because $H>0$. It cannot be of mode~2 specified by $k\in \mathbb{N}^*$ and $l\in \mathbb{N}^*$ or mode~3 because, otherwise $T^{1}>h(\xi_{1}(z_\alpha)) > h(M)=H$ according to Lemmas~\ref{thm:mode_2_cond} and \ref{thm:mode_3_cond}. Therefore, $z_{\alpha}$ is of mode~1 or mode 2 specified by some $k\in \nit^*$ and 0.

For any $\alpha$ such that $z_{\alpha}$ is of mode 2 specified by some $k\in \nit^*$ and 0, it follows from Lemma~\ref{thm:mode_2_cond} that there exists some $w\in (0,T^{N}]$ satisfying the conditions in eq.\eqref{cond_2_f} so that $\alpha \in \cS^{N}_{2}(w, T^{-N};k,0)$. Clearly, $\bar{\alpha}$ is in $\cS^{N}_{2}(T^{N},T^{-N};1,0)$. Thus, $\bar{z}$ is of mode 2 specified by 1 and 0. Lemma~\ref{thm:mode_2_cond}, eq.\eqref{flow_2_d} and eq.\eqref{flow_2_c} imply $\bar{x}_{1} = T^{N} - \frac{T^{N} - h(\bar{\xi}_{1})}{1+a(\bar{\xi}_{1})}<T^{N}$, $\bar{y}_{1}=M^{-N}$ and $\bar{\xi}_{1}=F^{-1}_{1}(T^{N})$.

2. According to Lemma~\ref{thm:mode_1_cond}, $z_{\alpha}$ is of mode~1 if and only if $\alpha\in \cS^{N}_{1}$, i.e. $\alpha=\underline{\alpha}$ or $\alpha^{1}\leq H$. In particular, an SA strategy $s$ is in $\cS^{N}_{1}$ if and only if $s \leq H$. Fix a strategy $\alpha\in \cS^{N}_{1}$. Then, $x_{1}(z_\alpha)=T^{N}$, $y_{1}(z_\alpha)=M^{-N}$ and $\xi_1(z_\alpha)=M$ by eq.\eqref{flow_1_a}. Since $x_{1}(z_\alpha)>\bar{x}_{1}$ and $y_{1}(z_\alpha)=\bar{y}_{1}$, it follows from Lemma~\ref{lm:reply_0} that $U^{N}(\alpha,T^{-N})>U^{N}(\bar{\alpha},T^{-N})$. In other words, no strategy in $\cS^{N}_{1}$ is optimal. In addition, for all $s\in [0,H]$, $\xi(z_s) =M$, $x_1(z_s)=T^N$ and $y^j_0(z_s)=T^j$ for $0\leq j \leq N-1$, i.e. they are all constant in $s$.

Now consider an SA strategy $s \in (H,T^{N}]$. It is not in $\cS^{N}_{1}$, hence $z_s$ is of mode 2 specified by 1 and 0. The total weight on arc 1 at $z_s$ is $\xi_{1}(z_s)=F^{-1}_{1}(s)$ by Lemma~\ref{thm:mode_2_cond}. By abuse of notation, define two functions of $s$, $\xi_1$ and $x_1$, by $\xi_1(s)\triangleq\xi_1(z_s)$ and $x_1(s)\triangleq x_1(z_s)$. Since $F_{1}$ is strictly decreasing, a bijection $\theta$ can be defined from interval $(H, T^{N}]$, the domain of $s$, to interval $[F^{-1}_{1}(T^{N}), M)$, the domain of $\xi_{1}$, such that $\theta=F^{-1}_{1}$ and $\theta^{-1}=F_1$. Then, atomic player~$N$'s cost $U^{N}(s,T^{-N})$, as a function of $s$, can be written as a function $v$ of $\xi_{1}$ on $[F^{-1}_{1}(T^{N}),M)$: $v(\xi_{1})=U^{N}(\theta^{-1}(\xi_{1}),T^{-N})=U^{N}(F_{1}(\xi_{1}),T^{-N})$.

According to eq.\eqref{flow_2_d}, $x_{1}=T^{N}-\frac{F_{1}(\xi_{1})-h(\xi_{1})}{1+a(\xi_{1})}=T^{N}-M+\xi_{1}$. Then, for all $\xi_{1}\in [F^{-1}_{1}(T^{N}),M)$,
\begin{equation*}
v(\xi_{1}) = x_{1}c_{1}(\xi_{1})+(T^{N}-x_{1})\,c_{2}(M-\xi_{1}) = (T^{N}-M+\xi_{1})c_{1}(\xi_{1})+(M-\xi_{1})c_{2}(M-\xi_{1}).
\end{equation*}
Its derivative function is
\begin{align*}
v'(\xi_{1}) & = c_{1}(\xi_{1}) + (T^{N}-M+\xi_{1})\,c'_{1}(\xi_{1}) - c_{2}(M-\xi_{1}) - (M-\xi_{1})c'_{2}(M-\xi_{1}) \\
& = c'_{1}(\xi_{1}) \,( T^{N}-F_{1}(\xi_{1}) ) = c'_{1}(\xi_{1}) \,(T^{N}-s )\geq 0
\end{align*}
and equality holds if and only if $s=T^{N}$ or, equivalently, $\xi_{1}= F^{-1}_{1}(T^{N})$.

Therefore, $v(\xi_{1})$ attains its unique minimum on interval $[F^{-1}_{1}(T^{N}),M)$ at $F^{-1}_{1}(T^{N})$, and it is strictly increasing on $[F^{-1}_{1}(T^{N}),M)$. As a result, the trivial strategy $\bar{\alpha}$ is optimal, and it is the unique SA strategy that is optimal.

Let us show that $\bar{\alpha}$ is the unique optimal strategy. Recall that no strategy in $\cS^{N}_{1}$ is optimal, hence it is enough to show that no strategy in $\cS^{N} \!\setminus \! \cS^{N}_{1}$ other than $\bar{\alpha}$ is optimal. Given an arbitrary $\alpha\in \cS^{N} \!\setminus \! \cS^{N}_{1}$, suppose that it is in $\cS^{N}_{2}(w,T^{-N};k,0)$ for some $w\in (0,T^{N}]$ and $k\in \mathbb{N}^{*}$. According to Lemma~\ref{thm:mode_2_cond}, $\alpha$ is equivalent to SA strategy $w-(k-1) h(\eta_1)$, where $\eta_1=F^{-1}_{k}(w)$. If SA strategy $w-(k-1)h(\eta_1)<T^N$, then $\alpha$ is not optimal. If $w-(k-1)h(\eta_1)=T^N$, then $\alpha$ induces the same aggregate weight on arc~1 as $\bar{\alpha}$, i.e. $\eta_1=F^{-1}_{1}(T^{N})$. As a result, $w=T^{N}+(k-1)h(F^{-1}_{1}(T^{N}))$. On the one hand, $F^{-1}_{1}(T^{N})<M$ and, consequently, $h(F^{-1}_{1}(T^{N}))>h(M)=H>0$. It follows that $w=T^{N}+(k-1)\,h(F^{-1}_{1}(T^{N}))\geq T^{N}$, and equality holds if and only if $k=1$. On the other hand, $w \leq T^{N}$. Therefore, $k=1$ and $w=T^{N}$, and $\alpha$ is just $\bar{\alpha}$.

3. It is already shown that for $s\in [0,H]$, $\xi$, $x_1(z_s)$ and $y_1(z_s)$ are all constant in $s$. For $s\in (0,H]$, recall that $\xi_1(z_s)=F^{-1}(s)$, and $y_1(z_s)=M^{-N}$ by Lemma~\ref{thm:mode_2_cond}. Thus $\xi_1(z_s)$ is strictly decreasing in $s$, $y_1(z_s)$ is constant. Consequently, $x_1(z_s)=\xi_1(z_s)-y_1(z_s)$ is strictly increasing in $s$.
\end{proof}


\begin{lemma}[Nontrivial case]\label{lm:lin_2}
In the cases not treated in Lemmas~\ref{lm:all_best} and \ref{lm:lin_1}, one has the following.

1. Atomic player~$N$ has at least one optimal decentralization strategy.

2. If $\alpha\in \cS^{N}$ is optimal, then $z_\alpha$ can be of mode~3 specified by $l_{0}$, or of mode~2 specified by some $k\in \mathbb{N}^{*}$ and some $l\in \{0, 1, \ldots, l_{0} \}$.

3. For SA strategy $s\in [0,T^N]$, $\xi_1(z_s)$ and $x_1(z_s)$ are continuous and non-increasing in $s$, and $y_1(z_s)$ is continuous and non-decreasing in $s$.
\end{lemma}
The proof follows arguments similar to those in the previous proof, but much longer. It is omitted to save space.

\begin{proof}[\textbf{Proof of Theorem~\ref{thm:best}}]
The result follows from Lemmas~\ref{lm:all_best}, \ref{lm:lin_1}, and \ref{lm:lin_2}, which treat the three cases respectively. For the nontrivial case, Lemma~\ref{lm:lin_2} shows that $U(s,T^{-N})$ is a continuous function in SA strategy $s$ on $[0,T^N]$, hence a minimizer exists. 
\end{proof}

\begin{corollary}\label{prop:flow_mono}
Consider SA strategies $s\in SA^N=[0,T^N]$ and the corresponding CE $z_s$ in the induced game $\Gamma(s,T^{-N})$. Both $\xi_1(z_s)$ and $x_1(z_s)$ are non-increasing in $s$, whereas $y_1(z_s)$ is non-decreasing in $s$.
\end{corollary}
\begin{proof}
It is a straight forward corollary of Lemmas~\ref{lm:all_best}--\ref{lm:lin_2}.
\end{proof}

\begin{proof}[\textbf{Proof of Theorem~\ref{thm:stackelberg}}]
According to Theorem~\ref{thm:sym_single}, it is sufficient to prove
\begin{equation}\label{eq:stackel}
\inf_{x\in X^N} \Pi^N(x,T^{-N}) = \min_{s \in S\!A^{N}} U^N(s,T^{-N}),
\end{equation}
and the minimizer on the left hand side exists and is unique. 

First show that $\inf_{x\in X^N} \Pi^N(x,T^{-N}) \leq \min_{s \in S\!A^{N}} U^N(s,T^{-N})$. Indeed, for any $s\in SA^N$, recall that the aggregate flow of the deputies of atomic player~$N$ at $z_s$ is $x(z_s)=(x_r(z_s))_{r=1}^{2}$, and the flows of the players in $T^{-N}$ are $y(z_s)=((y^l_r(z_s))^{2}_{r=1})_{l=0}^{N-1}$. In Stackelberg game $S\!\Gamma(T^N,T^{-N})$, by playing strategy $x(z_s)$, the CE of the induced composite congestion game $\Gamma_{x(z_s)}(T^{-N})$, $Z_{x(z_s)}$, is just $y(z_s)$. This is because at $(x(z_s),Z_{x(z_s)})$ and at $z_s$, the flows of the players in $T^{-N}$ satisfy the same equilibrium condition and such flows are unique. Thus,  $\Pi^N(x(z_s),T^{-N})=U^N(s,T^{-N})$ and, consequently, $\inf_{x\in X^N} \Pi^N(x,T^{-N}) \leq \min_{s \in S\!A^{N}} U^N(s,T^{-N})$.

Next show that $\inf_{x\in X^N} \Pi^N(x,T^{-N}) \geq \min_{s \in S\!A^{N}} U^N(s,T^{-N})$. For all decentralization strategy $s\in S\!A^N$, let $z'_s=(x(z_s),y(z_s))$ be the semi-aggregate flow induced by $z_s$ by considering only  the aggregate flow of atomic player~$N$'s deputies. Now in Stackelberg game $S\!\Gamma(T^N, T^{-N})$, for an arbitrary strategy $x\in X^N$ of the leader atomic player~$N$, if it induces a CE $Z_x$ in the composite congestion game $\Gamma_x(T^{-N})$ such that $(x, Z_x)=z'_s$ for some $s\in SA^N$, then $\Pi^N(x,T^{-N}) = U^N(s,T^{-N}) \geq \min_{t \in S\!A^{N}} U^N(t, T^{-N})$. Suppose that there exists $x\in X^N$ such that no $s\in SA^N$ satisfy $(x,Z_x)=z'_s$. Let us show that $\Pi^N(x,T^{-N})$ is not lower than $U^N(\bar{\alpha}, T^{-N})$, the cost to atomic player~$N$ when she plays the trivial decentralization strategy $\bar{\alpha}$. In other words, such a strategy cannot be strictly better than $x^N$, with $x^N$ being $N$'s flow at the CE $x$ of congestion game $\Gamma(T)$. 

For the sake of simplicity, denote $z_{\bar{\alpha}}$ by $\bar{z}$, $x_r(z_{\bar{\alpha}})$ by $\bar{x}_r$, $y^j_r(z_{\bar{\alpha}})$ by $\bar{y}^j_r$, $y_r(z_{\bar{\alpha}})$ by $\bar{y}_r$, and $\xi_r(z_{\bar{\alpha}})$ by $\bar{\xi}_r$; besides, denote by $\xi_r$ the total weight on arc $r$ at $(x,Z_x)$, and by $y^j=(y^j_r)_{r=1}^{2}$ the flow of atomic player~$j$ or that of the nonatomic players in $T^{-N}$ at $(x,Z_x)$. 

If $\xi_1=\bar{\xi}_1=\hat{\xi}$, then $\Pi^N(x,T^{-N})= U^N(\bar{\alpha}, T^{-N})=T^N c(\hat{\xi})$. If $\xi_1= \bar{\xi}_1\neq \hat{\xi}$, then by the proof of Lemma~\ref{lm:equiv_decentralization}, $x_1=\bar{x}_1$ and $y^j_1=\bar{y}^j_1$, $0 \leq j \leq N-1$. Consequently, $\bar{z}=(x,Z_x)$, a contradiction. Thus one should only consider the case that $\xi_1 \neq \bar{\xi}_1$. 

According to Corollary~\ref{prop:flow_mono}, $x_1(z_s)$ and $\xi_1(z_s)$ are non-increasing and continuous in $s$ while $y_1(z_s)$ is non-decreasing and continuous in SA decentralization strategy $s\in [0,T^N]$. In particular, the maximum (resp. minimum) of $\xi_1(z_s)$ and $x_1(z_s)$ are attained at $s=0$ (resp. $s=T^N$), i.e. when atomic player~$N$ plays the nonatomic decentralization strategy $\underline{\alpha}$ (resp. the trivial decentralization strategy $\bar{\alpha}$). Since no $s$ satisfies $(x,Z_x)=z'_s$, $x_1\notin [\bar{x}_1,x_1(z_{\underline{\alpha}})]$ (because otherwise a contradiction can be obtained by arguments similar to the proof of Lemma~\ref{lm:equiv_decentralization}). But $x_1(z_{\underline{\alpha}})=T^N$, thus $x_1$ cannot be greater than this. Therefore $x_1 < \bar{x}_1$. If $\xi_1 > \bar{\xi}_1$, then for all $j\in \{0,\ldots, N\!-\!1\}$, $y^j_1 \leq \bar{y}^j_1$. Indeed, for $j=0$, since $c_1(\bar{\xi}_1)<c_{1}(\xi_1)\leq c_{2}(\xi_2) < c_2(\bar{\xi}_2)$, $\bar{y}^0_1=T^0 \geq y^0_1$. For each $1 \leq j \leq N-1$, $\bar{y}^j_1$ is the solution of the following equations in $w$: 
\begin{align}
\text{either  }\, w< T^j \, \text{ and }\, c_{1}(\bar{\xi}_1)+ w\,c'_1(\bar{\xi}_{1}) = c_{2}(\bar{\xi}_2)+ (T^j-w) \,c'_2(\bar{\xi_{2}}),\label{tildey1}\\
\text{or  } \, w=T^j \, \text{ and } \, c_{1}(\bar{\xi}_1)+ w\,c'_1(\bar{\xi}_{1}) \leq c_{2}(\bar{\xi}_2)+ (T^j-w) \,c'_2(\bar{\xi}_{2}), \label{tildey2}
\end{align}
while $y^j_1$ is the solution of these equation in $w$:
\begin{align}
\text{either  }\, w< T^j \, \text{ and }\, c_{1}(\xi_1)+ w\,c'_1(\xi_{1}) = c_{2}(\xi_2)+ (T^j-w) \,c'_2(\xi_{2}), \label{y1}\\
\text{or  } \, w=T^j \, \text{ and } \, c_{1}(\xi_1)+ w\,c'_1(\xi_{1}) \leq c_{2}(\xi_2)+ (T^j-w) \,c'_2(\xi_{2}). \label{y2}
\end{align}
If $\bar{y}^j_1$ satisfies eq.\eqref{tildey1}, i.e. $c_{1}(\bar{\xi}_1)+ \bar{y}^j_1\,c'_1(\bar{\xi}_{1}) = c_{2}(\bar{\xi}_2)+ (T^j-\bar{y}^j_1) \,c'_2(\bar{\xi_{2}})$, since $\xi_1> \bar{\xi_1}$, $c_{1}(\xi_1)+ \bar{y}^j_1\,c'_1(\xi_{1}) > c_{2}(\xi_2)+ (T^j-\bar{y}^j_1) \,c'_2(\bar{\xi_{2}})$. Thus, if $y^j_1$ satisfy either eq.\eqref{y1} or eq.\eqref{y2}, then $y^j_1 < \bar{y}^j_1$. If $\bar{y}^j_1$ satisfies eq.\eqref{tildey2}, then $\bar{y}^j_1=T^j \geq y^j_1$.

Therefore $y_1 \leq \bar{y}_1$. But $x_1 < \bar{x}_1$, hence $\xi_1 < \bar{\xi}_1$, contradictory to the hypothesis that $\xi_1>\bar{\xi}_1$. This proves that $\xi_1 \leq \bar{\xi}_1$ and consequently $y_1 \geq \bar{y}_1$. Let $\Delta\xi = \bar{\xi}_1 - \xi_1$ and $\Delta x = \bar{x}_1- x_1$. Then $\Delta x \geq \Delta \xi$. Now let us show that $\Pi^N(x,T^{-N}) > U^N(\bar{\alpha}, T^{-N})$.  

At $\bar{z}$, $c_1(\bar{\xi}_1) \leq c_2(\bar{\xi}_2)$ and, according to eq.\eqref{cond_atom}, $c_1(\bar{\xi}_1)+\bar{x}_1 c_1'(\bar{\xi}_1) \leq c_2(\bar{\xi}_2)+ \bar{x}_2 c_2'(\bar{\xi}_2)$. Therefore, for all $s\in [0,\Delta \xi]$,
\begin{align}
c_1(\bar{\xi}_1-s) & \leq c_2(\bar{\xi}_2+s), \label{xx1}\\
c_1(\bar{\xi}_1-s)+(\bar{x}_1-s) c_1'(\bar{\xi}_1-s) &\leq c_2(\bar{\xi}_2+s)+ (\bar{x}_2 +s) c_2'(\bar{\xi}_2+s), \label{xx2}
\end{align}

Finally,
\begin{align*}
& \Pi^N(x,T^{-N}) - U^N(\bar{\alpha}, T^{-N})\\
& = [ x_1 c_1(\!\xi_1\!) + x_2 c_2 (\!\xi_2\!) ] \!-\! [\bar{x}_1 c_1(\!\bar{\xi}_1\!) + \bar{x}_2 c_2 (\!\bar{\xi}_2\!) ] \\
& = [x_2 c_2 (\!\xi_2\!)\! -\! \bar{x}_2 c_2 (\!\bar{\xi}_2\!)] \!-\! [\bar{x}_1 c_1(\!\bar{\xi}_1\!) \! -\! x_1 c_1(\!\xi_1\!)]\\
& = [(\bar{x}_2+\Delta \xi) c_2 (\bar{\xi}_2+\Delta \xi) - \bar{x}_2c_2 (\bar{\xi}_2) + (\Delta x - \Delta \xi) c_2 (\bar{\xi}_2+\Delta \xi)]\\
& ~~~~ - [\bar{x}_1 c_1(\bar{\xi}_1)  - ( \bar{x}_1 - \Delta \xi) c_1(\bar{\xi}_1-\Delta \xi) + (\Delta x-\Delta \xi ) c_1(\bar{\xi}_1-\Delta \xi) ]\\
& = \int_{0}^{\Delta \xi} c(\bar{\xi}_2+s)+(\bar{x}_2 + s) c_2'(\bar{\xi}_2+s) \, \text{d} s - \int_{0}^{\Delta \xi} c_1(\bar{\xi}_1-s)+(\bar{x}_1-s) c_1'(\bar{\xi}_1-s)\, \text{d} s\\
& ~~~~ +  (\Delta x-\Delta \xi ) ( c_2 (\bar{\xi}_2+\Delta \xi) -  c_1(\bar{\xi}_1-\Delta \xi) ]  \geq 0,
\end{align*}
because of eq.\eqref{xx1} and eq.\eqref{xx2}.

Eq.\eqref{eq:stackel} is proved. Meanwhile, an equilibrium strategy of player~$N$ in the Stackelberg game, i.e. a minimizer of the left hand side of eq.\eqref{eq:stackel}, is also found: $(x_r(z_\alpha))_{r=1}^2$, where $\alpha$ is an optimal unilateral decentralization of player~$N$. 
\end{proof}

\begin{proof}[\textbf{Proof of Theorem~\ref{thm:social_2}}]
For the sake of simplicity, denote $z_{\bar{\alpha}}$ by $\bar{z}$, $x_r(z_{\bar{\alpha}})$ by $\bar{x}_r$, $y^j(z_{\bar{\alpha}})$ by $\bar{y}^j$, and $\xi_r(z_{\bar{\alpha}})$ by $\bar{\xi}_r$.

In the nonatomic case, atomic player~$N$'s all decentralization strategies, including the trivial one (not decentralizing), result in the same outcome. Hence, the other players' costs and the social cost never change after the decentralization.

Now consider the trivial case and the nontrivial case. According to Lemmas~\ref{lm:lin_1}, \ref{lm:lin_2} and their proofs,  $\bar{z}$ is of mode 2.  Then $c_1(\bar{\xi}_1) < c_2(\bar{\xi}_2)$. 

It is sufficient to prove the result for all SA strategies $s\in [0,T^N)$. Corollary~\ref{prop:flow_mono} states that $\xi_1(z_s)$ and $x_1(z_s)$ are non-increasing in $s$ and, in particular, strictly decreasing in a neighborhood around $T^N$, while $y_1(z_s)$ are non-decreasing in $s$. Fix an $s\in [0,T^N)$, and denote $x_r(z_s)$ simply by $x_r$, $y^j_r(z_s)$ by $y^j_r$, and $\xi_r(z_s)$ by $\xi_r$. Then, $\xi_{1} > \bar{\xi}_{1}$ and $x_1> \bar{x}_1$. By the same argument used in the proof of Theorem~\ref{thm:stackelberg}, one can show that for all $j \in I\setminus \{0,N\}$, $y^j_1 \leq \bar{y}^j_1$. 
%

Let us compare the costs of the players in $I \!\setminus\! \{N\}$ and the social cost at $z_s$ with those at $\bar{z}$. 

For the nonatomic players in $T^0$, the fact that $\xi_{1} \geq \bar{\xi}_{1}$ immediately implies that $u^{0}(z)=c_{1}(\xi_{1}) \geq c_{1}(\bar{\xi_{1}})=u^{0}(\bar{z})$. Equality holds if and only if $\bar{\xi}_1=\xi_1$, which is impossible.

For $j \in I \! \setminus  \!\{0,N\}$ such that $c_{1}(\bar{\xi}_{1}) + \bar{y}^{j}_{1} c'_{1}(\bar{\xi}_{1}) = c_{2}(\bar{\xi}_{2})+ \bar{y}^{j}_{2} c'_{2}(\bar{\xi}_{2})$ and, consequently, $\bar{y}^{j}_{1} c'_{1}(\bar{\xi}_{1}) > \bar{y}^{j}_{2} c'_{2}(\bar{\xi}_{2})$ because $c_{1}(\bar{\xi}_{1}) < c_{2}(\bar{\xi}_{2})$. Let $B$ be a constant such that $\bar{y}^{j}_{1} c'_{1}(\bar{\xi}_{1}) > B > \bar{y}^{j}_{2} c'_{2}(\bar{\xi}_{2})$. Then, for all $s \in ( \bar{\xi}_1- \bar{y}^{j}_{1}, M - \bar{y}^{j}_{1}]$ and all $t\in [-\bar{y}^{j}_{2}, \bar{\xi}_2 - \bar{y}^{j}_{2})$,
\begin{equation}\label{eq:B}
 \bar{y}^{j}_{1} c'_{1}(\bar{y}^{j}_{1}+s) > B > \bar{y}^{j}_{2} c'_{2}(\bar{y}^{j}_{2}+t).
\end{equation}

It follows from the relation $\xi_{1} > \bar{\xi}_{1}$ that $\xi_{1} - \bar{y}^{j}_{1} >\bar{\xi}_1- \bar{y}^{j}_{1}$ and $\xi_{2} - \bar{y}^{j}_{2} < \bar{\xi}_2 - \bar{y}^{j}_{2}$. Therefore,
\begin{align*}
& [ \bar{y}^{j}_{1} c_{1}(\xi_{1}) + \bar{y}^{j}_{2} c_{2}(\xi_{2}) ] - [ \bar{y}^{j}_{1} c_{1}(\bar{\xi}_{1})+\bar{y}^{j}_{2} c_{2}(\bar{\xi}_{2}) ] \\
& = [ \bar{y}^{j}_{1} c_{1}(\bar{y}^{j}_{1} + \xi_{1} - \bar{y}^{j}_{1}) + \bar{y}^{j}_{2} c_{2}(\bar{y}^{j}_{2} + \xi_{2} - \bar{y}^{j}_{2}) ] - [ \bar{y}^{j}_{1} c_{1}(\bar{y}^{j}_{1} + \bar{y}^{-j}_{1}) + \bar{y}^{j}_{2} c_{2}(\bar{y}^{j}_{2} + \bar{y}^{-j}_{2}) ] \\
& = \bar{y}^{j}_{1} [ c_{1}(\bar{y}^{j}_{1} + \xi_{1} - \bar{y}^{j}_{1})-c_{1}(\bar{y}^{j}_{1} + \bar{y}^{-j}_{1}) ] - \bar{y}^{j}_{2} [ c_{1}(\bar{y}^{j}_{2} + \bar{y}^{-j}_{2}) - c_{2}(\bar{y}^{j}_{2} + \xi_{2} -\bar{y}^{j}_{2}) ]\\
& = \int^{\xi_{1} - \bar{y}^{j}_{1}}_{\bar{\xi}_1- \bar{y}^{j}_{1}} \!\bar{y}^{j}_{1} c'_{1}(\!\bar{y}^{j}_{1}\!+\!s\!) \text{d}s\! -\! \int^{\bar{\xi}_2 - \bar{y}^{j}_{2}}_{\xi_{2} - \bar{y}^{j}_{2}} \bar{y}^{j}_{2} c'_{2}(\!\bar{y}^{j}_{2}\!+\!t\!)\text{d}t  \!>\! [\xi_{1}\!-\!\bar{\xi}_{1}] B \!- \! [\bar{\xi}^{-j}_{2}\!-\!\xi_{2}]\!B=\!0,
\end{align*}
where the inequality is due to eq.\eqref{eq:B}, and
\begin{align*}
& [ y^{j}_{1} c_{1}(\xi_{1})+y^{j}_{2} c_{2}(\xi_{2}) ] - [ \bar{y}^{j}_{1} c_{1}(\xi_{1})+\bar{y}^{j}_{2} c_{2}(\xi_{2}) ] \\
& = [ y^{j}_{1}-\bar{y}^{j}_{1} ] c_{1}(\xi_{1}) + [ y^{j}_{2} - \bar{y}^{j}_{2} ] c_{2}(\xi_{2}) 
= [y^{j}_{1}-\bar{y}^{j}_{1}] [c_{1}(\xi_{1})-c_{2}(\xi_{2})]
\geq 0
\end{align*}
because $y^{j}_{1}\leq \bar{y}^{j}_{1}$ and $c_{1}(\xi_{1}) < c_{2}(\xi_{2})$. As a result,
\begin{align*}
 & u^{j}(z_s)-u^{j}(\bar{z}) = [ y^{j}_{1}\,c_{1}(\xi_{1})+y^{j}_{2}\,c_{2}(\xi_{2}) ] - [ \bar{y}^{j}_{1}\,c_{1}(\bar{\xi}_{1})+ \bar{x}^{j}_{2}\,c_{2}(\bar{\xi}_{2}) ]\\
& = [ y^{j}_{1} c_{1}(\!\xi_{1}\!) \!+ \!y^{j}_{2} c_{2}(\!\xi_{2}\!) ]\! -\! [ \bar{y}^{j}_{1} c_{1}(\!\xi_{1}\!) +\! \bar{y}^{j}_{2} c_{2}(\!\xi_{2}\!) ]\!+\! [ \bar{y}^{j}_{1} c_{1}(\!\xi_{1}\!)\!+\!\bar{y}^{j}_{2} c_{2}(\!\xi_{2}\!) ] \! \\
& ~~~-\! [ \bar{y}^{j}_{1} c_{1}(\!\bar{\xi}_{1}\!) \!+\! \bar{y}^{j}_{2} c_{2}(\!\bar{\xi}_{2}\!) ] \\
& >0.
\end{align*}

For $j \in I \! \setminus  \!\{0,N\}$ such that $c_{1}(\bar{\xi}_{1}) + \bar{y}^{j}_{1} c'_{1}(\bar{\xi}_{1}) < c_{2}(\bar{\xi}_{2})+ \bar{y}^{j}_{2} c'_{2}(\bar{\xi}_{2})$, $\bar{y}^{j}_{1}\!=\!T^{j}\!\geq \!y^{j}_{1}$. Recall that $c_{2}(\xi_{2})\geq c_{1}(\xi_{1})$ and $\xi_{1} > \bar{\xi}_{1}$. Therefore,
\begin{equation*} 
u^{j}(z)-u^{j}(\bar{z}) = [\,y^{j}_{1}\,c_{1}(\xi_{1})+y^{j}_{2}\,c_{2}(\xi_{2})\,] - T^{j}\,c_{1}(\bar{\xi}_{1}) \geq T^{j}\,c_{1}(\xi_{1})-T^{j}\,c_{1}(\bar{\xi}_{1})> 0.
\end{equation*} 


Finally, consider the social cost. Since $\bar{z}$ is of mode 2 specified by 1 and $l$, it follows from eq.\eqref{cond_atom} that $c_{1}(\bar{\xi_{1}}) + \bar{x}_{1}\,c'_{1}(\bar{\xi}_{1}) = c_{2}(\bar{\xi}_{2}) + \bar{x}_{2}\,c'_{2}(\bar{\xi}_{2})$, and $c_{1}(\bar{\xi_{1}}) + \bar{y}^{j}_{1}\,c'_{1}(\bar{\xi}_{1}) = c_{2}(\bar{\xi}_{2}) + \bar{y}^{j}_{2}\,c'_{2}(\bar{\xi}_{2})$ for $1\leq j \leq l$ if $l\geq 1$. Summing these $l+1$ equations leads to $(l+1) c_{1}(\bar{\xi}_{1}) + [ \bar{\xi}_{1}-(M^{-N}-T^{[l]})]\,c'_{1}(\bar{\xi}_{1})= (l+1) c_{2}(\bar{\xi}_{2}) + \bar{\xi}_{2}\,c'_{2}(\bar{\xi}_{2})$. Consequently, $c_{1}(\bar{\xi}_{1}) + \bar{\xi}_{1}\,c'_{1}(\bar{\xi}_{1})=c_{2}(\bar{\xi}_{2}) + \bar{\xi}_{2}\,c'_{2}(\bar{\xi}_{2}) + l[\,c_{2}(\bar{\xi}_{2})-c_{1}(\bar{\xi}_{1})]+(M^{-N}-T^{[l]})\,c'_{1}(\bar{\xi}_{1})$. Since $c_{1}(\bar{\xi}_{1})\leq c_{2}(\bar{\xi}_{2})$ and $l \geq 0$, there exists a constant $D>0$ such that $c_{1}(\bar{\xi}_{1}) + \bar{\xi}_{1}\,c'_{1}(\bar{\xi}_{1}) \geq  D\geq   c_{2}(\bar{\xi}_{2}) + \bar{\xi}_{2}\,c'_{2}(\bar{\xi}_{2})$. According to Assumption~\ref{assp:base_chap3}, $c_{1}$ and $c_{2}$ are both strictly increasing while $c'_{1}$ and $c'_{2}$ are non-decreasing. Hence, for any $s \in (\bar{\xi}_{1}, M]$ and any $t \in [0, \bar{\xi}_{2})$,
\begin{equation}\label{eq:lemma0_2}
c_{1}(s) +s\,c'_{1}(s) > D >   c_{2}(t) + t\,c'_{2}(t).
\end{equation}

Since $\xi_1>\bar{\xi}_{1}$, eq.\eqref{eq:lemma0_2} implies that
\begin{align*}
& CS(z_s) - CS(\bar{z}) \\
& = [ \xi_1\,c_{1}(\xi_1) + (M-\xi_2)\, c_{2}(M-\xi_2) ] - [\bar{\xi}_{1}\,c_{1}(\bar{\xi}_{1}) + (M-\bar{\xi}_{1})\, c_{2}(M-\bar{\xi}_{1}) ]\\
& = [ \xi_1\,c_{1}(\xi_1) - \bar{\xi}_{1}\,c_{1}(\bar{\xi}_{1}) ] - [(M-\bar{\xi}_{1})\, c_{2}(M-\bar{\xi}_{1})  - (M-\xi_1) c_{2}(M-\xi_1) ]\\
& 
= \int^{t}_{\bar{\xi}_{1}} [c_{1}(s) +s c'_{1}(s) ] \text{d}s - \int^{M-\bar{\xi}_{1}}_{M-\xi_1} [c_{2}(t) +t c'_{2}(t)] \text{d}t\\
& > (\xi_1-\bar{\xi}_{1}) D - (M-\bar{\xi}_{1}-M+\xi_1) D \\
&= 0.
\end{align*}
\end{proof}

\paragraph{Acknowledgments} I am very grateful to Professor Sylvain Sorin for our discussions on the subject of the paper and for his comments on its previous versions. I am also indebted to the anonymous referees whose comments and suggestions have helped hugely to improve the paper.



\end{document}